\documentclass[11pt]{article}
\usepackage[margin=1.15in]{geometry}
\usepackage{amsmath,amssymb,amsthm,mathtools}
\usepackage{hyperref}
\usepackage{enumitem}
\usepackage{algorithm,caption}
\usepackage[many]{tcolorbox}
\usepackage{nicefrac}
\usepackage{authblk}
\usepackage{microtype}
\usetikzlibrary{graphs}

\newcommand{\X}{\mathcal{X}} 
\newcommand{\C}{\mathcal{C}} 
\newcommand{\KRM}{K_{\mathrm{RM}}} 

\newcommand{\LAS}{\mathrm{LAS}}
\newcommand{\OPT}{\mathrm{OPT}}
\newcommand{\STAB}{\mathrm{STAB}}
\newcommand{\QSTAB}{\mathrm{QSTAB}}

\newcommand{\proj}{\mathrm{proj}}

\theoremstyle{plain}
\newtheorem{theorem}{Theorem}
\newtheorem{lemma}{Lemma}
\newtheorem{proposition}{Proposition}
\newtheorem{corollary}{Corollary}
\theoremstyle{definition}
\newtheorem{definition}{Definition}
\newtheorem{remark}{Remark}
\newtheorem{observation}{Observation}

\title{\textbf{Grouped Color Deletion, Lasserre Exactness and Clique-Sum Locality for Rainbow Matching}}

\author{Georgios Stamoulis\thanks{georgios.stamoulis@maastrichtuniversity.nl}}
\affil{Department of Advanced Computing Sciences \\ Maastricht University, The Netherlands}

\date{}
\begin{document}
\date{}

\maketitle

\begin{abstract}
In this paper we study the \emph{rainbow matching} problem: given an edge-colored graph, find a maximum matching that uses at most one edge of each color class. Rainbow matchings correspond to stable sets in the \textit{augmented graph} $H$ obtained from the line graph by completing each color class into a clique. For a hereditary graph class $\X$, we introduce and study the parameter $\kappa_\X$ to be the minimum number of color classes whose deletion places the \textit{residual} augmented graph in $\X$.

We show that this parameter has two complementary  flavors. Firstly, from a polyhedral point of view, we show that if $\X$ is \emph{uniformly rank-$r$ exact} (that is, every graph in $\X$ has a stable set polytope description by rank inequalities of right hand side $\leq r$), then deleting $k$ color classes to obtain a residual augmented graph that belongs to $\X$ implies \textit{exactness} of the Lasserre relaxation at level $k+r$. This yields, in particular, exactness at level $k+1$ for deletion to \emph{perfect} residual graphs, and exactness at level $k+r$ for deletion to $h$-\textit{perfect} residual graphs of bounded odd-hole rank $r$.

Our second result is structural. We show that the right object in this case is the \emph{color-intersection graph} $\Gamma$ where vertices are color classes, and we connect two vertices if the corresponding colors are incident on a common vertex. We show that $\Gamma$ impacts the topology of the augmented/conflict graph $H$ as follows: articulation colors in $\Gamma$ induce clique-sum decompositions in $H$, so residual obstructions for \textit{clique-sum-local} hereditary classes $\X$ (such as chordal, perfect, bipartite) are embedded in individual blocks. This shows that we can test membership of the residual graph in these target classes in a blockwise manner. As a consequence, for any such target class whose recognition can be done in polynomial time, we give an exact dynamic programming algorithm for computing the deletion parameter when the color-intersection graph has blocks of bounded size.

Finally we show that, in the chordal case, given a chordalizing color set the rainbow matching problem can be solved exactly by branching over the deleted color classes and solving chordal residual instances. We also show that computing the deletion parameter is \textbf{NP}-hard already in the chordal target case but it becomes FPT for classes $\X$ characterized by a set of forbidden induced subgraphs of bounded size.
\end{abstract}

\newpage
\section{Introduction}
Matching problems with side constraints are a very natural generalization of the matching problem and they have been studied extensively from combinatorial, complexity and polyhedral point of view. A particularly natural example is the \emph{rainbow matching} \cite{DBLP:journals/jct/Drisko98, DBLP:journals/ejc/AharoniBCHS19}: we are given an edge-colored graph and we are asked if there exists a matching of size at least $m$ such that all colors that appear are distinct. We may also impose upper bounds on how many edges per color we may use, leading to bounded-color matchings and other related grouped packing problems \cite{DBLP:journals/tcs/MastrolilliS14}. 


Rainbow Matching is \textbf{NP}-hard \cite{DBLP:journals/jacm/ItaiRT78} and \textbf{APX}-hard already on properly colored paths \cite{pfender2014complexity}. Constant-factor approximation algorithms exist based on combinatorial \cite{pfender2014complexity} or Linear Programming \cite{DBLP:conf/mfcs/Stamoulis14} techniques. The latter relaxations perform worse than the former local search arguments as they suffer from relatively large \emph{integrality gaps} even on very simple instances such as a bichromatic $C_4$ with alternating colors, and it is not clear when, and how, such gaps may be reduced, something that would imply better approximability properties \cite{DBLP:journals/disopt/KelkS19}.

In this paper we continue the study of the Rainbow Matching problem but we take a more structural viewpoint. Instead of asking how badly a particular relaxation can fail in the worst case (a notion captured by large integrality gaps), we are asking \textit{what graph-theoretic parameters influence the exactness rank of lift-and-project hierarchies?} Our starting point is the standard conflict graph formulation of the rainbow matching where edges of the original graph become vertices, and two such vertices are adjacent if their corresponding edges cannot be chosen together, either due to vertex or to color constraints. Rainbow matching is therefore exactly a stable set problem on a graph obtained from a line graph by completing each color class into a clique. We call this graph the \emph{augmented graph} and we denote it by $H$.

\subsection{Our contributions}
We introduce a \textit{grouped deletion} parameter for Rainbow Matching. Given an edge-colored graph and its conflict/augmented graph $H$, for every hereditary graph class $\mathcal{X}$ we define $\kappa_{\mathcal{X}}$ as the minimum number of color classes whose deletion places the corresponding \textit{residual} augmented graph (after the deletion of the vertices corresponding to these colors) in class $\mathcal{X}$. We show that this parameter has both polyhedral as well as structural and algorithmic interpretations.

Our first  result is a polyhedral \textit{rank} theorem and is related primarily to the geometry of its linear/SDP relaxation: it shows that the Lasserre relaxation becomes exact (integral) when a grouped deletion set exists. On a high level, the underlying mechanism is the following: suppose that after deleting $k$ color classes the residual augmented graph belongs to a hereditary class $\X$ whose stable set polytope is described by valid rank inequalities of right-hand side at most $r$. Then, a point in level $(k+r)$ of the Lasserre hierarchy contains enough \textit{moment} information to be decomposed along the deleted coordinates while still retaining an $r$-round feasible Lasserre object on each conditioned branch allowing us to combine Lasserre decomposition results with bounded-rank residual exactness, giving a general residual exact-class theorem. The strongest rank-$1$ instances come from perfect graphs. We also show how to handle $h$-perfect graphs of \textit{bounded odd-hole rank} and we also comment on the behavior of the integrality gap when we do not have the full budget of $k+r$ rounds that guarantees integrality. 

Our second  result is structural with  algorithmic consequences: we study the structure of the parameter $\kappa_\X$ in hereditary target classes $\X$ that are closed under \textit{clique-sums}: This class includes all classes the minimal forbidden induced subgraphs of which are
connected and have no clique cutset. 
A clique-sum decomposes a graph by ``gluing'' subgraphs along a common clique, with no edges between the non-shared parts. Chordal graphs, in particular, constitute a very intuitive such class and we use this case, for expositional purposes, to present the dynamic-programming recurrence concretely.

The algorithmic side requires knowing when membership, in a particular target class, after deletion can be checked locally. This is not visible in the conflict/augmented graph. We prove the following chain of results: first we form the \textit{color intersection} graph $\Gamma$ where color classes are vertices in $\Gamma$, and two vertices are connected if there exists a vertex in the original graph that is incident to edges of both colors.  We show that \textit{articulation colors} in $\Gamma$ induce exact clique cutsets in the augmented graph H, leading us to a clique-sum decomposition and to a structural result that says that induced residual obstructions (such as induced residual holes in the chordal case) are visible only within individual and distinct blocks. 
These results lead us to an exact dynamic programming algorithm for computing the parameter when $\Gamma$ has bounded block size.  We also show that given such a set of color classes $F$ the deletion of which places the residual augmented graph in some hereditary target class $\X$ with polynomial time recognition, we can solve maximum rainbow matching exactly by branching over the chosen edge in each deleted color class and solving the compatible residual instance. Together, these  imply fixed-parameter tractability for Rainbow Matching in these target classes  when the color-intersection graph has bounded block size and the relevant color classes are bounded.


Finally, we show that computing the chordal deletion number is \textbf{NP}-hard already in the chordal case, while for hereditary target classes characterized by finitely many forbidden induced subgraphs of bounded size, the grouped deletion problem is FPT.

Thus, the contribution, and the conceptual message, of our manuscript is  that the same grouped structural deletion parameter turns out to have two distinct but compatible flavors: \textit{(i)} In the residual augmented graph it defines a small vertex set $F$ on which Lasserre can be conditioned. \textit{(ii)} in the color-intersection graph, the exact same colors form the link where the instance decomposes into blocks allowing for a direct DP. In the clique-sum closed targets and the bounded block size of the color-intersection graph case, the same deleted set $F$ bridges the geometric and combinatorial contributions by extracting the optimal Rainbow Matching and also certifying that Lasserre level $|F|+1$ is exact.  Thus, the same deletion set $F$ controls at the same time, Lasserre rank, block locality, and exact optimization.


\begin{table}
    \centering
    \begin{tabular}{c|c|c}
         Result & Assumption  & Consequence \\
         \hline 
         \hline
         Lasserre Exactness  & $H - S(F) \in \X, \X ~r\mbox{-rank exact} $ & level-($|F|+r$) exact  \\
         Clique Sum Locality & $\X$ is clique sum local & blockwise membership \\
         DP & size of blocks($\Gamma$) $\leq b$ & compute  $\kappa$ in $2^{\mathcal{O}(b)} \cdot \mathrm{poly}$ \\
         Branching Algorithm & chordalizing $F$, bounded $|\C_i|$ & compute RM is FPT \\ 
         Complexity  & Chordal target $H$ & \textbf{NP}-hard \\
         Complexity & finite, bounded induced obstructions & FPT \\
    \end{tabular}
    \caption{Summary of our results}
    \label{tab:placeholder}
\end{table}

\begin{remark}
A set of colors $F$ whose deletion places the residual augmented graph in some easy (from stable set point of view) class $\X$ cannot be ignored because the optimal solution for the Rainbow Matching can \emph{still} use colors from $F$.  $F$ can be thought of as a boundary layer: once we have identified $F$  we should still account for the choices made inside the deleted color classes $F$ and solve the compatible residual instance, see Theorem \ref{supplied-chordal-set}. On the polyhedral side, the mere existence of such set already implies exactness of a low level of the Lasserre relaxation for the whole instance. 
\end{remark}


\subsection{Related work}
Deletion parameters to reach a target graph class $\X$ have been extensively studied in the graph-modification and parameterized complexity literature. In classical deletion to $\X$ settings  we may delete vertices (or edges) individually to reach a (hereditary) graph class $\X$. Cai’s framework \cite{DBLP:journals/ipl/Cai96} gives the canonical FPT result for hereditary properties characterized by finitely many forbidden induced subgraphs. Later work studies specific target classes such as chordal \cite{DBLP:journals/algorithmica/Marx10} or, more generally, perfect graphs \cite{DBLP:journals/tcs/HeggernesHJKV13}. However, in our case, we delete \textit{entire color classes} of the primal graph that correspond to a prescribed group of vertices in the conflict graph, at unit cost.

Thus  $\kappa_\X$ can be best viewed as a \textit{group deletion/modulator to a hereditary class} $\X$, see for example the $\mathcal G$-modular cardinality/restricted modular partition framework of Lafond and Luo \cite{DBLP:conf/mfcs/LafondL23}, where the complexity is controlled by decomposing the graph into a bounded number of structured parts that are manipulated as units.  A second  analogy comes from \emph{grouped or color-based deletion} problems on edge-colored graphs see for example \cite{DBLP:journals/mst/MorawietzGKS22}. 
The idea is precisely to choose a small set of colors and delete \emph{all} associated edges. 

Rainbow matching has also been studied extensively from a parameterized point of view, but the parameter of interest there was the  natural solution size $m$. For this case, FPT results and a cubic kernel was given in \cite{gupta2019parameterized} which was soon after improved to a quadratic kernel in \cite{DBLP:journals/algorithmica/GuptaRSZ20}. Moreover, \cite{DBLP:conf/swat/FominGK024} provides a point of view similar in spirit but complementary to the one in this manuscript: there the authors studied independent stable set in matroid frameworks, where partition matroids capture naturally rainbow constraints and rainbow matching appears as a line-graph special case. They provided algorithms on graph classes such as chordal graphs under linear matroids under representations.  Finally, \cite{DBLP:conf/soda/BessyBTW23} use Rainbow Matching to design kernels for various packing problems. 

Finally, a large body of work concerns exactness of various families of linear or semidefinite relaxations \cite{laurent2003comparison}. For example, and again without intention to exhaust the  literature,  Lov\'asz theta function is a powerful SDP relaxation of the maximum stable set problem, and it defines the ``theta body'',  a convex set bounding the stable set polytope  which is tight for perfect graphs \cite{thetabody, DBLP:journals/mp/BianchiENT17, chvatal1975certain, lovasz1991cones}. On the other hand, bounded-width structure can force hierarchy ``collapse'' to the integral hull: this appears, for example, in Sherali–Adams \cite{DBLP:journals/disopt/BienstockO04, DBLP:journals/siamcomp/CarbonnelRZ22} and Lasserre exactness results \cite{wainwright2004treewidth} parameterized by the treewidth. 
Our result is orthogonal: we do not insist on bounding the width of the entire conflict graph but we show that we can still obtain exactness from a small grouped ``modulator'' to a hereditary residual class whose stable-set polytope admits a bounded-rank description.

\subsection{Graph Theoretic Preliminaries}
All graphs are finite and simple. For a graph $G$ and a subset $U$ of its  vertices, $G[U]$ is the subgraph induced by $U$. A graph is \textit{chordal} if it has no induced cycle of length at least four. A \textit{perfect elimination ordering} (PEO) of a chordal graph is an ordering of the vertices in which every vertex is simplicial in the graph induced by itself and the later vertices. A graph $X$ is perfect if every induced subgraph $Y$ satisfies $\chi(Y) = \omega(Y)$, where $\chi(Y)$ is the chromatic number of $Y$ and $\omega$ is its maximum clique.  It is $h$-perfect if its stable set polytope is described by nonnegativity, clique inequalities, and odd-hole inequalities \cite{hperfect}. Perfect graphs are also $h$-perfect. The stable set polytope of a graph $X$ is denoted by $\STAB(X)$ and is the convex hull of all incidence vectors of stable sets. The \emph{fractional stable set relaxation} of $X$ is
\[
\mathrm{FRAC}(X)=\Big\{x\in [0,1]^{V(X)}: ~~x_u+x_v\leq 1 \text{ for every } \{u,v\}\in E(X)\Big\},
\]
and the \emph{clique relaxation} is
\[
\QSTAB(X)=\Big\{x\in [0,1]^{V(X)} : ~~x(Q) = \sum_{u \in Q} x_u \leq 1 \text{ for each clique } Q \mbox { of } X \Big\}.
\]

Let $(G,\C)$ be an edge-colored graph and for an edge $e$ we denote its unique color by $c(e)\in\C$. Its \textit{conflict graph} $H$ has one vertex for each edge of $G$, and two vertices are adjacent if the corresponding edges share an endpoint or have the same color. Thus rainbow matchings in $(G,\C)$ are exactly maximum stable sets in $H$. The \textit{line graph} $L(G)$ of $G$ is the graph whose vertices are the edges of  $G$ with adjacency defined by shared endpoints. We obtain the \emph{augmented} graph  by completing each color class into a clique. The \textit{color-intersection graph} $\Gamma$  has one vertex for each color, with two colors adjacent if some edge of one color shares an endpoint with some edge of the other. For a set $F$ of colors $S(F) \subseteq V(H)$ denotes the set of conflict-graph vertices corresponding to edges whose colors lie in $F$. A \textit{block} (biconnected component) of a graph is a maximal 2-connected subgraph, together with isolated vertices and bridge edges in the standard block decomposition. An articulation vertex is a vertex whose removal increases the number of connected components. A \emph{clique cutset} is a clique the removal of which increases the number of connected components. A \emph{clique sum} of two graphs is obtained by identifying in both graphs a clique of teh same size and ``gluing'' the two graphs together along that clique.  The Rainbow Matching (RM) problem can be formally stated as follows:
\begin{definition}[Rainbow Matching]
Given a simple, non-directed graph $G=(V,E)$ where each edge $e \in E$ has a color $c(e) \in \C$, $\C = \{\C_1, \dots \C_k\}$ is a partition of the edges into $k$ color classes,  find the maximum matching $\mathcal{M}$ in $G$ with $|\mathcal{M} \cap C_i | \leq 1$  $\forall i \in [k]$.
\end{definition} 

We can formulate RM as a linear program in the straightforward way: associate a variable $x_e$ for each $e \in E$, and maximize $\sum_{e} x_e$ subject to standard matching and color budgets constraints
\[
\sum_{e: v \in e} x_e \leq 1, \quad \forall v \in V(G), \quad \mbox{ and } \sum_{e \in \C_i} x_e \leq 1, \quad \forall \C_i \in \C.
\]
Dropping the integrality constraints $x_e \in \{0,1\}$ gives us the linear relaxation of the problem which has integrality gap equal to 2 even for very restrictive graphs. We denote the linear relaxation for RM above as $K_{\mathrm{RM}}$.  Every degree or color constraint of the rainbow-matching formulation induces a valid conflict inequality in its conflict graph $H$, so every integer feasible RM solution corresponds to a stable set in $H$.

\subsection{Tools from the Lasserre Hierarchy}
The Lasserre hierarchy has been introduced by Jean Lasserre in \cite{DBLP:conf/ipco/Lasserre01, DBLP:journals/siamjo/Lasserre01}, and since then has found numerous applications in combinatorial optimization. Our notation follows closely the lecture notes \cite{rothvoss2013lasserre}, see also \cite{laurent2003comparison}. 

Let $K= \{x \in \mathbb{R}^n \mid A x \geq b\}$ be a polyhedron, where $A \in \mathbb{R}^{m \times n}$ and $b \in \mathbb{R}^m$. (We assume the bounding constraints $0 \leq x_i \leq 1$ are included in this system). To capture joint probabilities, Lasserre considers vectors $y \in \mathbb{R}^{2^{[n]}}$ indexed by subsets of variables. For an integer $t \geq 1$, we define two types of moment matrices for $y$:

\begin{description}
\item[The Moment Matrix]: $M_t(y)$ is a symmetric matrix whose rows and columns are indexed by sets $I, J \subseteq [n]$ with $|I|, |J| \leq t$. Its entries are defined by $(M_t(y))_{I,J} := y_{I \cup J}$.

\item[The Slack Matrices]: For each $\ell \in [m]$ corresponding to the linear constraint $\sum_{i=1}^n A_{\ell i} x_i \geq b_\ell$, the slack moment matrix $M_t^\ell(y)$ is indexed by $I, J \subseteq [n]$ with $|I|, |J| \leq t$. Its entries are defined by $(M_t^\ell(y))_{I,J} = \sum_{i=1}^n A_{\ell i} y_{I \cup J \cup \{i\}} - b_\ell y_{I \cup J}$.
\end{description}

\begin{definition}
The $t$-th level of the Lasserre hierarchy, denoted $\LAS_t(K)$, is the set of vectors $y \in \mathbb{R}^{2^{[n]}}$ satisfying: 
\[
M_t(y) \succeq 0, \quad M_t^\ell(y) \succeq 0 \quad \forall \ell \in [m], \quad y_\emptyset = 1.
\]
\end{definition}
For rainbow matching we will use $K = \KRM$ and for the rank-closure lemma we use $K=\mathrm{FRAC}(H)$, $H$ is the conflict graph of the input graph $G$. 
The primary geometric object of interest is the projection of $\LAS_t(K)$ back onto the original variable space:
\[
\LAS_t^{\proj}(K) = \Big\{(y_{\{1\}}, \dots, y_{\{n\}}) \mid y \in \LAS_t(K) \Big\} 
\]

We will crucially utilize the following \textit{Decomposition Lemma} \cite{karlin2011integrality, rothvoss2013lasserre} which intuitively says that if every valid integral solution can choose at most $k$ elements from a given subset $S$, then any Lasserre fractional point at level $t \geq k$ can be perfectly decomposed into a convex combination of points that are strictly integral on $S$:

\begin{lemma}\label{decompositionlemma}
Let $0 \leq k \leq t$ and $y \in \LAS_t(K)$. Let $S \subseteq [n]$ be a subset of indices such that every integral solution $x \in K \cap \{0,1\}^n$ satisfies $|supp(x) \cap S| \leq k$. Then, $y$ can be expressed as a convex combination of feasible Lasserre solutions at level $t-k$ that are perfectly integral on $S$. Formally:
\[
y \in \mathrm{conv}\Big\{z \mid z \in \LAS_{t-k}(K), ~z_i \in \{0,1\} ~ \forall i \in S\Big\}.
\]
\end{lemma}

With the Decomposition Lemma in our hands, we will extract its most crucial polyhedral consequence for the stable set problem (on the conflict graph $H$). In the canonical fractional stable set relaxation the total fractional mass distributed across an induced subgraph $H[U]$ can exceed its true combinatorial independence number, $\alpha(H[U])$. However, if we would have the budget to condition on the vertices of $U$ up to this bound then the Lasserre hierarchy would systematically eliminate this fractional ``cheating''. This leads to our primary polyhedral tool:  the \textit{\textbf{rank-closure lemma}}, which establishes that exactly $r$ rounds of Lasserre are sufficient to globally enforce the rank inequality of any induced subgraph with independence number $r$. For RM we will work with $\KRM$ but we make the next rank closure lemma independent of the formulation. The following results are folklore. 

\begin{lemma}
\label{rank-closurelemma}
Let $H$ be a graph, let $K \subseteq[0,1]^{V(H)}$ be a polyhedron such that $K \subseteq \mathrm{FRAC}(H)$. 
Set $r =\alpha(H[U])$ for $U \subseteq V(H)$. If $y \in \LAS_r(K)$ and $x_v =y_{\{v\}}$ denotes its singleton projection, then $x(U) = \sum_{v \in U} x_v \leq r$.
\end{lemma}

\begin{proof}
We apply Lemma~\ref{decompositionlemma} with the set $S =U$, with the same integer $k =r$, and with level $t =r$. Since every integral stable set of $H$ uses at most $\alpha(H[U])=r$ vertices from $U$,
the hypotheses of Lemma~\ref{decompositionlemma} are satisfied. Hence 
\[
y=\sum_{\zeta}\lambda_\zeta\, y^{(\zeta)},
\]
where each $\lambda_\zeta \geq 0$, $\sum_\zeta \lambda_\zeta=1$, each
$y^{(\zeta)}\in LAS_0(K)=K$, and each $y^{(\zeta)}$ is integral on $U$.

Let $x^{(\zeta)}$ denote the singleton projection of $y^{(\zeta)}$. 
Since $y^{(\zeta)} \in K \subseteq \mathrm{FRAC}(H)$ the singleton projections $x^{(\zeta)}$ satisfy $x_u^{(\zeta)} + x_v^{(\zeta)} \leq 1$ for every edge $\{u,v\} \in E(H)$.  Since $x^{(\zeta)}$ is integral on $U$, the set $I_\zeta = \Big\{u\in U : x^{(\zeta)}_u = 1\Big\} $ is a stable set of the induced subgraph $H[U]$. Therefore $|I_\zeta|\leq \alpha(H[U])=r$. Equivalently, $ x^{(\zeta)}(U) \leq r$ for every $\zeta$. Now average over the convex decomposition:
\[
x(U) =  \sum_{u\in U} x_u =
\sum_{u\in U}\sum_\zeta \lambda_\zeta x^{(\zeta)}_u
=
\sum_\zeta \lambda_\zeta x^{(\zeta)}(U)
\leq \sum_\zeta \lambda_\zeta r = r.
\]
This proves the claim.
\end{proof}

\begin{corollary}
\label{clique-closure}
Let $Q$ be a clique in the conflict graph $H$. Then every singleton projection $x$ of a point $y\in \LAS_1(\KRM)$ satisfies $x(Q)\leq 1$.
\end{corollary}

\begin{proof}
Because $Q$ is a clique of the conflict graph, every two distinct vertices of $Q$ are conflicting. Hence every integral stable set of $H$ uses at most one vertex of $Q$. Equivalently, $\alpha(H[Q])=1$. Applying Lemma~\ref{rank-closurelemma} with $U = Q$, we obtain $x(Q)\le \alpha(H[Q])=1$.
\end{proof}

\section{Residual Exact Classes for Rainbow Matching}
The previous discussion showed us that level-$1$ Lasserre conditions enforce clique inequalities, and higher levels allow us to  condition on bounded subsets of variables. We now lift these algebraic tools into a more global structural framework by asking \textit{if deleting a small set of color classes from a rainbow matching instance leaves the residual augmented graph with a simpler polyhedral description, how many rounds of the hierarchy are required in order to capture the exact integral hull?}  The following ``meta-theorem'' demonstrates that the required \textit{Lasserre rank} is strictly bounded by the size of the deleted set plus the maximum right-hand side of the rank inequalities defining the residual stable set polytope. In the following, since the vertices of the conflict graph $H$ are in bijection with the edges of the original graph $G$, we will slightly abuse notation and identify a vertex of $H$ with its corresponding edge of $G$. We start with a necessary definition:
\begin{definition}
Let $\X$ be a hereditary class of graphs, and let $r \geq 1$. We say that $\X$ is \emph{uniformly rank-$r$ exact} if for every $J\in\X$, $\exists$ a family of subsets $\mathcal U(J)\subseteq 2^{V(J)}$ s.t:
\[
\STAB(J) = 
\Big\{x\in\mathbb R_{\geq 0}^{V(J)} : x(U)\leq \alpha(J[U])~~~ \text{ for all }~ U\in\mathcal U(J)\Big\},
\]
and $\alpha(J[U]) \leq r$ for each $U\in\mathcal U(J)$.
\end{definition}

\begin{theorem}
\label{lasserre-rainbow-theorem}
Let $\mathcal X$ be a hereditary graph class that is uniformly rank-$r$ exact. Let $(G,\C)$ be a colored graph, $H$ its conflict graph and $\KRM$ the natural LP relaxation for RM. For  $F\subseteq \C$, by $S(F)$ we denote the set of vertices of $H$ corresponding to edges of colors in $F$. If $|F|=k$, and $H-S(F)\in\X$, then
\[
\LAS_{k+r}(\KRM)^{\proj}=\STAB(H).
\]
\end{theorem}

\noindent
{\textbf{Proof idea.}} We condition on the deleted color classes via Lemma \ref{decompositionlemma} obtaining integral branches on $S(F)$. Fixing one branch selects a set $M_F$ of deleted color edges and forces every incompatible undeleted edge to disappear. The remaining compatible residual graph lies in $\X$, so Lemma \ref{rank-closurelemma} together with uniform rank-$r$ exactness implies that the restricted solution belongs to the stable set polytope of the residual graph. Recombining this residual stable-set decomposition with $M_F$ gives a convex decomposition into stable sets of $H$.

\begin{proof}
The inclusion $\STAB(H)\subseteq \LAS_{k+r}(\KRM)^{\mathrm{proj}}$ is immediate. For the reverse inclusion, let $y\in \LAS_{k+r}(\KRM)$ be a feasible Lasserre level-$(k+r)$ vector and let $x$ be its singleton projection. Since every integral rainbow matching uses at most one edge from each deleted color class, every integral feasible point uses at most $k$ vertices from $S(F)$. By Lemma \ref{decompositionlemma}, we can write $y$ as a convex combination of points in $\LAS_r(\KRM)$ that are integral on $S(F)$. Since $\STAB(H)$ is convex, it suffices to consider one such component. Thus, from now on, assume that $y\in \LAS_r(\KRM)$ and that $x$ is integral on $S(F)$. Let $M_F =\{e\in S(F):~x_e=1\}$ be the set of selected deleted-color edges, and define
\[
E' = \Big\{e\in V(H)\setminus S(F): e \text{ is nonadjacent in } H \text{ to every } f\in M_F\Big\}.
\]
We claim that $x_e=0$ for every $e\in V(H)\setminus (S(F)\cup E')$. Indeed, such a vertex $e$ is adjacent in $H$ to some $f\in M_F$. Since $f\in S(F)$ but $e\notin S(F)$ this conflict cannot arise from belonging to the same deleted color class and must therefore come from a shared endpoint in the original graph. Hence the base LP contains the degree constraint $x_e+x_f\leq 1$, and since $x_f=1$ (because $f \in M_F$), we obtain $x_e = 0$.

Now let $J=H[E']$. Since $E'\subseteq V(H)\setminus S(F)$, the graph $J$ is an induced subgraph of $H-S(F)$ and therefore $J\in\X$ because $\X$ is hereditary. Let $U\in \mathcal U(J)$. Every integral feasible rainbow matching is a stable set of $H$, hence uses at most $\alpha(J[U])$ vertices from $U$. Since $\alpha(J[U])\le r$ and $y\in \LAS_r(\KRM)$, Lemma \ref{rank-closurelemma} yields $x(U)\leq \alpha(J[U])$. Thus $x|_{E'}$ satisfies every defining inequality of $\STAB(J)$, and therefore $x|_{E'}\in \STAB(J)$.

Let $x|_{E'}=\sum_{s=1}^t \mu_s \chi^{I_s}$ where each $I_s$ is a stable set of $J$ and the $\mu_s$ forms a convex combination. For each $s$, the set $M_F\cup I_s$ is a stable set of $H$: the set $M_F$ is stable because $x|_{S(F)}$ is integral and feasible for the base constraints, each $I_s$ is stable in $J\subseteq H$, and by definition every vertex of $E'$ is nonadjacent to every vertex of $M_F$. Moreover, all coordinates outside $S(F)\cup E'$ disappear. So $x=\sum_{s=1}^t \mu_s \chi^{\,M_F\cup I_s}$, and so $x\in \STAB(H)$. This proves $\LAS_{k+r}  (\KRM)^{\mathrm{proj}}=\STAB(H)$.
\end{proof}

Rank-$1$ instances arise from perfect graphs, for which clique inequalities suffice. This exactness result follows immediately as a corollary of Theorem~\ref{lasserre-rainbow-theorem} by setting $r=1$. The chordal case is a particularly useful subclass, not exclusively by the stronger polyhedral consequences but, as we will see in the subsequent sections, by its structural and algorithmic corollaries. 

 To demonstrate that the above result generalizes to higher-rank polyhedral descriptions, we consider $h$-perfect graphs: the stable set polytope requires odd-hole inequalities in addition to clique constraints and non-negativity inequalities. By bounding the length of the residual odd holes, we strictly bound the right-hand side of the associated rank inequalities yielding Lasserre exactness at higher levels: 

\begin{corollary}
\label{cor:h-perfect-rainbow}
Let $F\subseteq \C$ with $|F|=k$. Suppose that $H-S(F)$ is $h$-perfect and that every odd hole of $H-S(F)$ has length  $\leq 2r+1$. Then $\LAS_{k+r}(\KRM)^{\proj}=\STAB(H)$.
\end{corollary}

\begin{proof}
We apply Theorem~\ref{lasserre-rainbow-theorem} with $\X$ equal to the class of $h$-perfect graphs whose odd holes have length at most $2r+1$. Let $J\in \mathcal \X$. By definition the polytope $\STAB(J)$ is described by the following inequalities: \textit{(i)} \textit{nonnegativity inequalities} $x_v\ge 0$ for all $v\in V(J)$, \textit{(ii)} \textit{clique inequalities} $x(Q)\le 1$ for every clique $Q\subseteq V(J)$, \textit{(iii)} \textit{odd-hole inequalities} $x(V(C))\le \frac{|V(C)|-1}{2}$ for every odd hole $C$ of $J$. We show that every defining inequality is a stable-set rank inequality with right-hand side at most $r$. First, for a clique $Q$, we have $\alpha(J[Q])=1$, so the clique inequality is exactly the rank inequality $x(Q)\le \alpha(J[Q])=1$. Second, let $C$ be an odd hole of $J$, say of length $|V(C)|=2\ell+1$. Then the odd-hole inequality is $x(V(C))\le \ell$. Since $C$ is an induced odd cycle $\alpha(C)=\ell$. Thus this inequality is precisely the rank inequality $x(V(C))\leq \alpha(J[V(C)])$. By assumption, every odd hole of $J$ has length at most $2r+1$, so every such $\ell$ satisfies $\ell\leq r$. Therefore every defining inequality of $\STAB(J)$ is of the form $x(U)\le \alpha(J[U])$ 
with right-hand side at most $r$. Thus the hypotheses of Theorem~\ref{lasserre-rainbow-theorem} are satisfied, and we conclude that $\LAS_{k+r}(\KRM)^{\proj}=\STAB(H)$.
\end{proof}

\subsection{Integrality Gap on Fewer Lasserre Rounds}
The following statement quantifies the integrality gap in cases where we may spend only a limited number of Lasserre rounds.  

\begin{theorem}
\label{hperfect-gap}
Let $(G,\C)$ be a rainbow matching instance, let $H$ be its conflict graph, let $F\subseteq \C$ be a set of $k$ colors. Let $J = H-S(F)$. Assume that $J$ is $h$-perfect. If $J$ is perfect then the previous result applies. Otherwise, let $2\ell+1$ be the length of a shortest odd hole in $J$. Then for every weight vector $w\in \mathbb R_{\ge 0}^{V(H)}$, we have
\[
\max \Big\{w^\top x: ~x\in \LAS_{k+1}(\KRM)^{\proj} \Big\}
\leq 
\left(1+\frac{1}{2\ell}\right)
\max\Big\{w^\top x: ~x\in \STAB(H) \Big\}.
\]
I.e., the integrality gap of $\LAS_{k+1}(\KRM)$ is at most $1+\frac{1}{2\ell}$.
\end{theorem}

Before we turn to the proof, we prove the following result: 

\begin{proposition}[Scaling bound for $h$-perfect graphs]\label{hperfect-scaling}
Let $J$ be an $h$-perfect graph, and assume that $J$ is not perfect. Let $2\ell+1$ be the length of a shortest odd hole in $J$, and define
\[
\beta_\ell = 1 + \frac{1}{2\ell}
\qquad\text{and}\qquad
\lambda_\ell=\frac{1}{\beta_\ell} = \frac{2\ell}{2\ell+1}.
\]
Then $\lambda_\ell\, P(J)\subseteq \STAB(J)$, where
\[
P(J) = \Big\{x\in \mathbb R_{\geq 0}^{V(J)} : x(Q) \leq 1 \text{ for every clique }Q\subseteq V(J)\Big\}
\]
is the clique relaxation of the stable set polytope.
\end{proposition}

\begin{proof}
Fix any $z\in P(J)$. We claim that $\lambda_\ell z\in \STAB(J)$. Since $J$ is $h$-perfect, it suffices to verify that $\lambda_\ell z$ satisfies: (i) nonnegativity, (ii)  every clique inequality, and (iii) every odd-hole inequality.

Nonnegativity is immediate, because $z\ge 0$ and $\lambda_\ell>0$. Next, let $Q$ be any clique of $J$. Since $z\in P(J)$, we have $z(Q)\le 1$. Because $\lambda_\ell\le 1$, it follows that $\lambda_\ell z(Q)\le z(Q)\le 1$, so every clique inequality is satisfied. It remains to check odd-hole inequalities. Let $C=(v_1,v_2,\dots,v_{2m+1})$ be an odd hole of $J$. Since $2\ell+1$ is the length of a shortest odd hole, we have $m\ge \ell$. Every edge of $C$ is a clique of size two, and therefore $z\in P(J)$ satisfies $z_{v_i}+z_{v_{i+1}}\leq 1$ for $i=1,\dots,2m+1$, where indices are taken cyclically modulo $2m+1$. Summing these  inequalities around the cycle gives $2\,z(C)\le 2m+1$, and hence $z(C)\le m+\frac{1}{2}$. Therefore
\[
\lambda_\ell z(C) ~\leq~ \frac{2\ell}{2\ell+1}\left(m+\frac12\right) 
~=~
\frac{2\ell m+\ell}{2\ell+1} ~\leq~ m,
\]
where the last inequality is equivalent to $\ell\le m$, which is true. Thus the odd-hole inequality for $C$ is satisfied by $\lambda_\ell z$.

We have shown that $\lambda_\ell z$ satisfies all nonnegativity, clique, and odd-hole inequalities of $J$. Since $J$ is $h$-perfect, these inequalities describe $\STAB(J)$, and therefore $\lambda_\ell z\in \STAB(J)$. This proves $\lambda_\ell P(J)\subseteq \STAB(J)$.
\end{proof}

Now we may prove the Theorem itself:
\begin{proof}
If $J$ is perfect, we have already established exactness. We therefore assume that $J$ is not perfect, and let $2\ell+1$ be the length of a shortest odd hole in $J$. Define 
\[
\OPT_{\LAS} = \max\big\{w^\top x: ~x\in \LAS_{k+1}(\KRM)^{\proj}\big\},
\qquad
\OPT = \max\big\{w^\top x: ~x\in \STAB(H)\big\}.
\]
Let us fix a point $y\in LAS_{k+1}(\KRM)$ whose singleton projection $x$ attains $\OPT_{\LAS}$. By the decomposition lemma applied to the deleted set $S(F)$, we can write $y=\sum_{s} \lambda_s\, y^{(s)}$,
where: \textit{(i)} $\lambda_s\geq 0$ for all $s$, \textit{(ii)} $\sum_s \lambda_s=1$,
\textit{(iii)} each $y^{(s)}\in \LAS_1(\KRM)$, and (iv) each $y^{(s)}$ is integral on the coordinates indexed by $S(F)$. Let $x^{(s)}$ denote the singleton projection of $y^{(s)}$. Fix one branch $s$. Define
\[
M_s = \big\{e\in S(F): x^{(s)}_e=1 \big\}, ~
E_s' = \big\{e\in V(H)\setminus S(F): e \text{ nonadjacent in } H \text{ to each } f\in M_s\big \}.
\]

We first show that $x^{(s)}_e=0$ for each $e\notin S(F)\cup E_s'$. Indeed, such a vertex $e$ is adjacent in $H$ to some $f\in M_s$. In the conflict graph, adjacency means that the corresponding edges of the original graph either share an endpoint or have the same color. Since $f\in S(F)$ and $e\notin S(F)$, the two edges cannot have the same deleted color. Hence they must share an endpoint. The base rainbow-matching formulation therefore contains a valid degree constraint $x_e+x_f\le 1$. Because $f\in M_s$, we have $x^{(s)}_f=1$, and hence $x^{(s)}_e=0$. Thus the support of $x^{(s)}$ is contained in $S(F)\cup E_s'$, and the objective splits as
\[
w^\top x^{(s)} = w(M_s)+ w^\top x^{(s)}|_{E_s'}.
\]
Now define a vector $z^{(s)}\in \mathbb R^{V(J)}$ by $ z^{(s)}_e= x^{(s)}_e$ if   $e\in E_s'$ and zero otherwise. We claim that $z^{(s)}\in P(J)$, where $P(J)$ is the clique relaxation of $J$.

First, $z^{(s)}\ge 0$ coordinatewise. Next, let $Q$ be any clique of $J$. Then $Q\cap E_s'$ is also a clique of $J$, hence a clique of the full conflict graph $H$. Since $y^{(s)}\in \LAS_1(\KRM)$, the level-1 clique-closure corollary gives $x^{(s)}(Q\cap E_s')\leq 1$. But $
z^{(s)}(Q)=x^{(s)}(Q\cap E_s')$, because $z^{(s)}$ agrees with $x^{(s)}$ on $E_s'$ and is zero outside $E_s'$. Therefore $z^{(s)}(Q)\leq 1$.
This proves $z^{(s)}\in P(J)$.

We now apply Proposition~\ref{hperfect-scaling} to $J$. Since $J$ is $h$-perfect and has shortest odd hole of length $2\ell+1$, we obtain
\[
\frac{2\ell}{2\ell+1}\, z^{(s)} \in \STAB(J).
\]
Restricting this vector to the induced subgraph $J_s =J[E_s']$ we get that 
\[
\frac{2\ell}{2\ell+1}\, x^{(s)}|_{E_s'} \in \STAB(J_s).
\]
Indeed, if $\frac{2\ell}{2\ell+1}\, z^{(s)}=\sum_t \mu_t \chi^{I_t}$
is a convex combination of stable sets $I_t$ of $J$, then restricting each $I_t$ to $E_s'$ yields a convex combination of stable sets of $J_s$.  Hence, writing $\alpha_w(J_s)$ for the maximum weight of a stable set in $J_s$, we obtain
\[
w^\top x^{(s)}|_{E_s'} \leq \left(1+\frac{1}{2\ell}\right)\alpha_w(J_s).
\]

On the other hand, every stable set of $J_s$ can be combined with $M_s$ to form a stable set of the full conflict graph $H$. Then the set $M_s$ is stable because $x^{(s)}|_{S(F)}$ is integral and feasible for the base rainbow matching constraints, and every vertex of $E_s'$ is nonadjacent in $H$ to every vertex of $M_s$ by construction. Therefore
\[
\OPT \geq w(M_s)+\alpha_w(J_s).
\]
Combining the previous inequalities yields
\[
w^\top x^{(s)}
=
w(M_s)+w^\top x^{(s)}|_{E_s'}
\le
w(M_s)+\left(1+\frac{1}{2\ell}\right)\alpha_w(J_s).
\]
Since $w(M_s)\geq 0$ and $1+\frac{1}{2\ell} \geq  1$, we  get that
\[
w^\top x^{(s)}
\leq
\left(1+\frac{1}{2\ell}\right)\bigl(w(M_s)+\alpha_w(J_s)\bigr)
\leq
\left(1+\frac{1}{2\ell}\right)\operatorname{OPT}.
\]
This bound holds for every branch $s$. Averaging over the convex decomposition gives us the desired result.
\end{proof}

\section{Chordal Topology, Deletion Parameters and Algorithms}
We now study the same grouped deletion parameter from a structural point of view. In the previous section deleted colors mattered because they defined a small set of coordinates on which Lasserre could condition. Here they matter because the exact same colors influence the topology of the color-intersection graph since articulation colors give clique separators in the conflict graph, and the relevant obstructions (induced holes in the chordal case) are not global but are visible  by the block structure of the color-intersection graph.


\subsection{Clique-sum decomposition from the color-intersection graph}

In the next  statements we isolate the structural reasons why the chordal case is algorithmically tractable: Lemma \ref{clique-sum-decomp}  shows that the conflict graph is assembled from the block subgraphs  $H_B$ by clique-sums along articulation colors of the color-intersection graph $\Gamma$. Lemma \ref{holes-cliquesum} is a very basic fact saying that induced holes cannot cross a clique-sum. Corollary \ref{locality} is the key consequence: every induced hole of the conflict graph \textit{is already contained in a single block subgraph $H_B$}. Thus we can check chordality of $H$ in a blockwise way. This statement is the structural input from the undeleted instance that will later exploit for the dynamic program.

\begin{lemma}
\label{clique-sum-decomp}
Let $B(\Gamma)$ denote the set of blocks of the color-intersection graph $\Gamma$. For each block $B\in B(\Gamma)$ we define $E(B)=\{e\in E(G): c(e)\in V(B)\}$, $H_B = H[E(B)]$. Then the conflict graph $H$ can be obtained as an iterated clique-sum of the block subgraphs $\{H_B: B\in B(\Gamma)\}$, where the gluing cliques are exactly the color cliques corresponding to articulation colors of $\Gamma$.
\end{lemma}

We first need some lemmata that would help us built our arguments towards the proof. We start with the result that shows that articulation colors induce clique cutsets.

\begin{lemma}
\label{articulation-cut}
Let $c$ be an articulation vertex of the color-intersection graph $\Gamma$, and let $A_1,\dots,A_t$ be the connected components of $\Gamma-c$. For each $i\in[t]$ we define $W_i =\{e\in E(G): c(e)\in V(A_i)\}$ and  $E_c =\{e\in E(G): c(e)=c\}$. Then:
\begin{enumerate}
    \item $E_c$ is a clique in the conflict graph $H$
    \item if $i\neq j$ then there are no edges in $H$ between $W_i$ and $W_j$.
\end{enumerate}
i.e., $E_c$ is a clique cutset of $H$ separating the sets $W_1,\dots,W_t$.
\end{lemma}

\begin{proof}
We first prove that $E_c$ is a clique in $H$. Any two distinct edges in $E_c$ have the same color $c$  and the vertices of the conflict graph are adjacent if the corresponding edges have
the same color. Therefore every two distinct vertices of $E_c$ are adjacent in $H$ and so so $E_c$ is a clique.

Now we fix two distinct indices $i\neq j$, and we let $e\in W_i, f\in W_j$. We claim that $e$ and $f$ are nonadjacent in $H$. Since $e\in W_i$ and $f\in W_j$, their colors satisfy $c(e)\in V(A_i), c(f)\in V(A_j)$ i.e., neither edge has color $c$ and the two colors belong to different connected components of $\Gamma-c$.

Suppose for contradiction that $e$ and $f$ are adjacent in $H$. Then by the definition of the conflict graph, either $e$ and $f$ have the same color, or $e$ and $f$ share a common endpoint in the original graph. The first case is impossible, because $c(e)\in A_i$ and $c(f)\in A_j$ with $i\neq j$, so the colors are distinct. In the second case, if $e$ and $f$ share an endpoint in the original graph, then the two colors $c(e)$ and $c(f)$ are adjacent in the color-intersection graph $\Gamma$. But this would give an
edge between a vertex of $A_i$ and a vertex of $A_j$ in $\Gamma-c$, contradicting that $A_i$ and $A_j$ are distinct connected components of $\Gamma-c$. Thus neither case can occur, and therefore $e$ and $f$ are nonadjacent in $H$. 
The set $E_c$ is therefore a clique whose removal separates the sets
$W_1,\dots,W_t$, so $E_c$ is a clique cutset of $H$.
\end{proof}

Now we prove the Clique-decomposition Lemma \ref{clique-sum-decomp} of the conflict graph: 

\begin{proof}
We argue by induction on the number of blocks of $\Gamma$. If $\Gamma$ has only one block, then $B(\Gamma)=\{\Gamma\}$, and by definition $E(\Gamma)=E(G), H_\Gamma=H$. So the statement is immediate.

Assume now that $\Gamma$ has at least two blocks. Consider the block-cut tree of $\Gamma$, and choose a leaf block $B$. Since $B$ is a leaf of the block-cut tree, there exists a unique articulation color $c\in V(B)$ through which $B$ is attached to the rest of the graph. Let
\[
E_B = E(B),
\qquad
E_{\mathrm{rest}} = \bigcup_{B'\in B(\Gamma)\setminus\{B\}} E(B').
\]
Define $H_1 =H[E_B], H_2 = H[E_{\mathrm{rest}}]$. We claim that \textit{(i)} $H = H_1 \cup H_2$, \textit{(ii)} $H_1 \cap H_2 = E_c$, and \textit{(iii)} there are no edges of $H$ between $V(H_1)\setminus E_c$ and $V(H_2)\setminus E_c$.

The first claim is true because every color of $\Gamma$ belongs either to the leaf block $B$ or to one of the remaining blocks, so every edge of the original graph belongs to $E_B\cup E_{\mathrm{rest}}$.  For the second claim, note that the only colors that can belong to both $B$ and another block are articulation colors. Since $B$ is a leaf block, the only such articulation color is $c$. Therefore the overlap of $E_B$ and $E_{\mathrm{rest}}$ consists exactly of the edges of color $c$, i.e., $E_B\cap E_{\mathrm{rest}}=E_c$. Hence $H_1\cap H_2=E_c$. For the third claim, every color appearing in $V(H_1)\setminus E_c$ lies in the unique connected component of $\Gamma-c$ that contains $B\setminus\{c\}$, while every color appearing in $V(H_2)\setminus E_c$ lies in the union of the other connected components of $\Gamma-c$.
Therefore Lemma~\ref{articulation-cut} applies and shows that there are no edges between $V(H_1)\setminus E_c$ and $V(H_2) \setminus E_c$.

Since $E_c$ is a clique by Lemma \ref{articulation-cut}, it follows that $H$ is the clique-sum of $H_1$ and $H_2$ over the clique $E_c$.

Finally, the graph $H_2$ is precisely the conflict graph obtained from the graph whose colors are the vertices of $\Gamma$ outside the leaf block $B$, and whose color-intersection graph has one fewer
block. By the induction hypothesis, $H_2$ is itself an iterated clique-sum of the block subgraphs $H_{B'}$ for $B'\neq B$. Gluing back $H_1=H_B$ over the clique $E_c$ gives an iterated clique-sum
representation of $H$ over all block subgraphs $\{H_B : B\in B(\Gamma)\}$.
\end{proof}

The following Lemma formalizes the intuitive notion that once two graphs are glued together, an induced hole cannot use vertices from both sides.

\begin{lemma}
\label{holes-cliquesum}
Let $G_1$ and $G_2$ be graphs such that $G = G_1 \cup G_2, K = V(G_1)\cap V(G_2)$, where $K$ is a clique, and there are no edges between $V(G_1)\setminus K$ and $V(G_2)\setminus K$. Then every induced cycle of $G$ of length at least $4$ is contained entirely
in $G_1$ or entirely in $G_2$.
\end{lemma}

\begin{proof}
Let $C$ be an induced cycle of $G$ of length at least $4$. Suppose for contradiction that $C$ uses vertices from both $V(G_1)\setminus K$ and $V(G_2)\setminus K$.

Because there are no edges between $V(G_1)\setminus K$ and $V(G_2)\setminus K$, every time the cycle passes from one side to the other it must do so through a vertex of $K$. In particular, the cycle must contain at least two vertices of $K$: with only one such vertex, the cycle could enter the clique separator from one side but could not return from the other side without using an edge directly between the two sides, which does not exist.

Thus $C$ contains two distinct vertices $u,v\in K$. Since $K$ is a clique, the edge $uv$ belongs to $G$. We claim that $u$ and $v$ are nonconsecutive on the cycle $C$. Indeed, if they were consecutive, then removing the edge $uv$ from the cycle would leave a path connecting a vertex of $V(G_1)\setminus K$ to a vertex of $V(G_2)\setminus K$ without using any additional vertex of $K$,
which is impossible because there are no edges between the two sides outside $K$. Therefore $u$ and $v$ are nonconsecutive vertices of $C$.

But then the edge $uv$ is a chord of $C$, contradicting the assumption that $C$ is induced. Hence every induced cycle of length at least $4$ lies entirely in one side.
\end{proof}

\begin{corollary}
\label{locality}
Every induced cycle of the conflict graph $H$ of length at least $4$ is contained in $H_B$ for some block $B$ of the color-intersection graph $\Gamma$. In particular, $H$ is chordal if and only if every block subgraph $H_B$ is chordal.
\end{corollary}

\begin{proof}
By Lemma \ref{clique-sum-decomp}, the graph $H$ is an iterated clique-sum of the block subgraphs $H_B$, with gluing cliques given by articulation colors. We prove the first statement by induction on the number of clique-sum operations. At each step, Lemma \ref{holes-cliquesum} shows that an induced cycle of length at least $4$ cannot use
vertices from both sides of the clique-sum decomposition  therefore every induced cycle of length at least $4$ is already contained entirely in one term.  Iterating this argument down the clique-sum decomposition tree, we conclude that every induced cycle of length at least $4$ is contained in one of the block subgraphs $H_B$. For the second statement, if every $H_B$ is chordal, then no $H_B$ contains an induced cycle of length at least $4$, and hence by the first part neither does $H$. Thus $H$ is chordal. The converse is immediate, since each $H_B$ is an induced subgraph of $H$, and chordality is hereditary property.
\end{proof}

\begin{remark}
Lemma \ref{holes-cliquesum} and Corollary \ref{locality} are states for chordal obstructions but there is nothing inherent about such induced cycles. We can replace such obstructions by any appropriate obstruction of a well-chosen target family, see Subsection \ref{localCliqueSum} and Definition \ref{cliquesumlocal}.
\end{remark}

We now state an elementary observation that will be used repeatedly: we identify exactly what happens in $H$ and in $\Gamma$ after we delete a set of colors. 

\begin{observation}
\label{residual}
Let $(G,\C)$ be an RM instance, let $H$ be its conflict graph, and let $\Gamma$ be its color-intersection graph. For any set of colors $F\subseteq \C$, let $(G^F,\C^F)$ denote the residual instance obtained by deleting all edges of colors in $F$. Then: \emph{(i)} the conflict graph of $(G^F,\C^F)$ is exactly $H-S(F)$, and \emph{(ii)} the color-intersection graph of $(G^F,\C^F)$ is exactly $\Gamma-F$.
\end{observation}

Finally, an easy lemma that it will be used repeatidely later on:

\begin{lemma}
\label{inheritanceblocks}
Let $X$ be any graph, and let $X'$ be an induced subgraph of $X$. Then every block of $X'$ is contained in some block of $X$.
\end{lemma}

\begin{proof}
Let $B'$ be a block of $X'$. If $B'$ is an isolated vertex or a bridge edge, then it is trivially contained in some block of $X$. Otherwise $B'$ is a maximal $2$-connected subgraph of $X'$. Since $X'$ is an induced subgraph of $X$,  any two internally vertex-disjoint paths in $X'$ between vertices of $B'$ are also such paths in $X$. Hence $B'$ is $2$-connected in $X$ as well. Therefore $B'$ is contained in some maximal $2$-connected subgraph of $X$, i.e. in some block of $X$.
\end{proof}

\subsection{Local clique-sum and hereditary classes beyond chordal targets}\label{localCliqueSum}
As mentioned briefly in the previous subsection, the chordal residual-locality argument is not specific to induced holes. It only uses heredity and the fact that the relevant minimal forbidden induced subgraphs cannot cross clique separators. In this subsection we formalize this mechanism.
We first define this clique-sum hereditary class more formally: 

\begin{definition}\label{cliquesumlocal}
Let $\X$ be a hereditary graph class. We say that $\X$ is \textbf{clique-sum local} if every graph that is minimal with respect to induced
subgraphs among graphs not belonging to $\mathcal{X}$ is connected and has no clique cutset.
\end{definition}

Notable examples of graph classes that fall into this category include Chordal graphs as it is demonstrated in the previous subsection specifically (the minimal forbidden induced subgraphs are induced cycles of length at least four), \textit{Perfect graphs} (where by the Strong Perfect Graph Theorem, minimal obstructions are odd holes and odd antiholes) and \textit{Bipartite graphs} (minimal forbidden induced subgraphs are odd cycles). In all these cases the minimal obstructions are connected and have no clique cutset. 

We note also that even if Perfect Graphs fall into this category, it does \textit{not} mean that all individual cases do as well, for example split and cographs do not. 

We now show that connected obstructions do not cross a clique-sum.
\begin{lemma}\label{nocrossobstructions}
Let $G_1$ and $G_2$ be graphs, let
\[
G = G_1 \cup G_2, \qquad K = V(G_1)\cap V(G_2),
\]
and assume that $K$ is a clique and that there are no edges between
$V(G_1)\setminus K$ and $V(G_2)\setminus K$. Then, if $J$ is an induced subgraph of $G$ that is connected and has no clique cutset, then we have that either $V(J)\subseteq V(G_1)$ holds or $V(J)\subseteq V(G_2)$ holds.
\end{lemma}

\begin{proof}
Suppose for the sake of contradiction that $J$ contains a vertex of $V(G_1)\setminus K$ and also a vertex of $V(G_2)\setminus K$. Since $J$ is connected and there are no edges between $V(G_1)\setminus K$ and $V(G_2)\setminus K$, the set $V(J)\cap K$ must be nonempty. Moreover, because $K$ is a clique in $G$, the set $V(J)\cap K$ is a clique in $J$.

We will show that $V(J)\cap K$ is a clique cutset of $J$. After deleting $V(J)\cap K$, no vertex of $V(J)\cap (V(G_1)\setminus K)$ is adjacent to a vertex of $V(J)\cap (V(G_2)\setminus K)$, simply because such edges do not exist in $G$. Thus $J - (V(J)\cap K)$ is disconnected, with vertices on both sides and so  $V(J)\cap K$ is a clique cutset of $J$, contradicting the assumption on $J$.
\end{proof}

\begin{theorem}\label{cliquesumlocal-theorem}
Let $\X$ be a clique-sum local hereditary graph class as defined above. Let $(G,\C)$ be an edge-colored graph, let $H$ be its conflict graph,
and let $\Gamma$ be its color-intersection graph. For each block $B$ of $\Gamma$, let $E(B) = \{e\in E(G): c(e)\in V(B)\}$, $H_B = H[E(B)]$. Then for every set of colors $F\subseteq \mathcal{C}$,
\[
H-S(F)\in \mathcal{X}
\quad\Longleftrightarrow\quad
H_B - S(F\cap V(B)) \in \mathcal{X}
\ ~\text{ for every block }~ B \in \mathcal{B}(\Gamma).
\]
\end{theorem}

\begin{proof}
The forward direction is immediate, simply because each graph $H_B-S(F\cap V(B))$ is an induced subgraph of $H-S(F)$ and $\X$ is hereditary.

For the converse, assume that
\[
H_B-S(F\cap V(B))\in\mathcal{X}
\quad\text{for each block }~ B\text{ of }\Gamma,
\]
but that $H-S(F) \notin \X$. Since $\X$ is a hereditary class there exists an induced subgraph $J$ of $H-S(F)$ that is minimal with respect to not belonging in $\X$. By the locality of the clique-sum, $J$ is connected and has no clique cutset. The residual instance obtained by deleting the colors in $F$ has conflict graph $H-S(F)$ and color-intersection graph $\Gamma-F$. 
Applying Lemma \ref{clique-sum-decomp} to this residual instance we see that the graph $H-S(F)$ is an iterated clique-sum of the block subgraphs corresponding to the blocks of $\Gamma-F$.
Repeated applications of the Lemma~\ref{nocrossobstructions} shows us that $J$ is contained entirely in the block subgraph associated with some block $B'$ of $\Gamma-F$. Now $B'$ is a block of the induced subgraph $\Gamma-F$. By Lemma \ref{inheritanceblocks} each block of an induced subgraph is contained in a block of the original graph, so $B'\subseteq B$ for some original block $B$ of $\Gamma$ from which we get that that $J \subseteq H_B - S(F\cap V(B))$. But notice that $J\notin \X$ and $\X$ is hereditary, contradicting the assumption that $H_B-S(F \cap V(B))\in\mathcal{X}$. Therefore $H-S(F)\in\mathcal{X}$.
\end{proof}

\subsection{Dynamic programming over the block-cut tree}\label{DP}

We now show that for any hereditary graph class $\X$ that is clique sum local (Definition \ref{cliquesumlocal}) the  deletion parameter $\kappa_\X$ becomes exactly computable when the block structure of the color-intersection graph is bounded. In the following we specialize once again, and only for clarity purposes, in the chordal case but all the arguments are identical in the more general case.

The main idea is that $\kappa_\X(G,\C)$ depends not on global graph $\Gamma$ alone but by \textit{how the non-chordality is distributed across the blocks of it}. The dynamic program depends on following three ingredients: 
\begin{itemize}
\item[\emph{(i)}] by component additivity property in Lemma \ref{kappa-components}, disconnected components of  $\Gamma$ can be solved independently, 
\item[\textit{(ii)}] Corollary \ref{locality} gives block-locality of holes in the remaining instance, 
\item[\textit{(iii)}] after deleting colors, Lemma \ref{residual-block-local-lemma} combines these ingredients and extends this to a residual block-locality statement: once the deletion status of articulation colors is fixed, chordality of the residual augmented graph can be tested blockwise on the original blocks of $\Gamma$. 
\end{itemize}

This is exactly a situation in which a dynamic program over the block-cut tree becomes possible.



\begin{remark}
Let $V_c, c \in \C$, denote the set of vertices in the original graph $G$ incident to some edge of color $c$, then $\Gamma$ is precisely the intersection graph of $\{V_c\}$ and thus bounded block size $b$ in $\Gamma$ means that the global color interaction can decompose into a tree of overlapping ``cores'', each involving $b$ colors at most.
\end{remark}

We start with the easy result that the parameter adds across connected components.
\begin{lemma}
\label{kappa-components}
Let $\Gamma_1,\dots,\Gamma_t$ be the connected components of the color-intersection graph $\Gamma$, and let $H_1,\dots,H_t$ be the
corresponding induced subgraphs of the conflict graph $H$. Then $H=H_1\sqcup \cdots \sqcup H_t$ is a disjoint union, and $\kappa_\X(G,\C)=\sum_{i=1}^t \kappa_i$,  where $\kappa_i$ denotes the chordal deletion number of the $i$th component instance.
\end{lemma}

\begin{proof}
If two colors lie in different connected components of $\Gamma$, then no edge of one color shares an endpoint with an edge of the other color. Therefore no conflict-graph edge can connect vertices
coming from different connected components of $\Gamma$, and $H$ is a disjoint union of the corresponding induced subgraphs. A graph is chordal if and only if each of its connected components is chordal. Hence a color set is  chordalizing if and only if its restriction to each connected component is chordalizing. The minimum size of such set is therefore the sum of the component-wise optima.
\end{proof}


\begin{lemma}
\label{residual-block-local-lemma}
Let $F\subseteq \C$ be any set of deleted colors. Then the residual augmented graph $H-S(F)$ is chordal
if and only if, for every original block $B$ of the color-intersection graph $\Gamma$, the local residual graph $H_B-S(F\cap V(B))$ is chordal.
\end{lemma}

\begin{proof}
The forward implication is straightforward  since each graph $H_B-S(F\cap V(B))$ is an induced subgraph of $H-S(F)$ and chordality is hereditary.

For the reverse direction, we assume that every graph $H_B-S(F\cap V(B))$ is chordal, but that $H-S(F)$ is not. Then $H-S(F)$ contains an induced cycle $C$ of length at least $4$. By Observation~\ref{residual}, the residual edge-colored instance obtained by deleting the colors in $F$ has conflict graph $H-S(F)$ and color-intersection graph $\Gamma-F$. Applying the earlier block-locality theorem to this residual instance, the induced cycle $C$ must be contained in the block subgraph associated with some block $B'$ of $\Gamma-F$. Now $B'$ is a block of an induced subgraph of $\Gamma$, so by Lemma \ref{inheritanceblocks}, $B'$ is contained in some original block $B$ of $\Gamma$. It follows that $C$ is contained in $H_B-S(F\cap V(B))$, contradicting the assumption that this graph is chordal. Hence $H-S(F)$ must be chordal.
\end{proof}

\begin{remark}
Lemma \ref{residual-block-local-lemma} immediately yields a useful global consequence. For each block $B \in \mathcal{B}(\Gamma)$, suppose that we choose a set of colors $F_B \subseteq V(B)$ such that the local residual augmented graph $H_B - S(F_B)$ is chordal. If we  define $F = \bigcup_{B \in \mathcal{B}(\Gamma)} F_B$, it then follows that $H - S(F)$ is chordal as well: for every block $B$ we have $F_B \subseteq F \cap V(B)$, so $H_B - S(F \cap V(B))$ is an induced subgraph of $H_B - S(F_B)$ and hence is chordal and by Lemma \ref{residual-block-local-lemma} we get that $H - S(F)$ is chordal and so any family of chordalizing color sets for the block subgraphs can be merged into a valid global chordalizing color set. Indeed if we define the local chordal deletion number of a block $B$ as  $\kappa_B = \min\{ |F_B| : F_B \subseteq V(B)\}$  and  $H_B - S(F_B)$ is chordal then $\kappa(G,\C) \leq \sum_{B \in \mathcal{B}(\Gamma)} \kappa_B$.

This upper bound is governed by the local conflict graphs $H_B$ instead by the topology of $\Gamma$ and this is essential because even a block of very few colors may contain a non-chordal local conflict graph. The dynamic program below we exploit exactly this form of locality combining block information once we have fixed the status of articulation colors.
\end{remark}

\begin{theorem}[Exact DP for bounded-block $\Gamma$]\label{DP-bounded-block-theorem}
Suppose that every block of the color-intersection graph $\Gamma$ has size at most $b$. Then the chordal deletion parameter $\kappa(G,\C)$ can be computed exactly in time $2^{\mathcal{O}(b)}\cdot \mathrm{poly}(|E|+|\C|)$.
\end{theorem}

\begin{proof}
By Lemma \ref{kappa-components} we can assume that $\Gamma$ is connected.

Let $T$ be the block-cut tree of $\Gamma$, rooted at an arbitrary block node $B_0$. The nodes of $T$
alternate between block nodes and articulation-color nodes. For a non-root block $B$, let $p(B)$ denote the unique parent articulation color connecting $B$ to the upper part of the tree. Let $A(B)$ to be the set of child articulation colors of $B$, and define
\[
I(B) = V(B)\setminus\bigl(A(B)\cup\{p(B)\}\bigr),
\]
i.e., the set of internal colors of $B$. For the root block $B_0$, we define $I(B_0) = V(B_0)\setminus A(B_0)$.  For each articulation color $a$, let $\mathcal B(a)$ denote the set of child blocks of $a$ in the rooted tree. The reason this state space is sufficient is precisely Lemma \ref{residual-block-local-lemma} which tells us that once the deletion status of the parent articulation color is fixed, the only remaining constraints are local to the current block and independent across distinct child subtrees. We now define the dynamic-programming states.

\noindent
\textbf{Articulation states.}
For an articulation color $a$ and  $\delta\in\{0,1\}$, let $DP_a(\delta)$ denote the minimum number of deleted colors in the subtree rooted at $a$, including the cost of $a$
itself if $\delta=1$, under the condition that $a$ is deleted when $\delta=1$ and kept when $\delta=0$.
 Once the status of $a$ is fixed, the subproblems on the child blocks are independent, so
\[
DP_a(\delta)=\delta+\sum_{B\in \mathcal B(a)} DP_B(\delta).
\]
Informally, \(DP_a(\delta)\) records the optimum cost in the subtree below an articulation color once
its deletion status is fixed, whereas \(DP_B(\delta)\) records the optimum cost in the sub-tree below a block under the boundary condition imposed by the parent articulation color.

\noindent
\textbf{Block states.}
For a non-root block $B$ and $\delta\in\{0,1\}$, we let $DP_B(\delta)$ denote the minimum number of deleted colors in the subtree rooted at $B$, excluding the possible cost
of the parent articulation color $p(B)$, under the condition that $p(B)$ has status $\delta$. To compute $DP_B(\delta)$, enumerate all subsets $X\subseteq A(B), Y\subseteq I(B)$.
Here $X$ specifies the child articulation colors deleted inside $B$, and $Y$ specifies the deleted internal colors. Define the set of deleted colors local to $B$ by
\[
D_{B,\delta,X,Y} =
Y\cup X\cup
\begin{cases}
\{p(B)\}, & \text{if } \delta=1,\\
\emptyset, & \text{if } \delta=0.
\end{cases}
\]
This local choice is feasible exactly when $H_B-S(D_{B,\delta,X,Y})$
is chordal. Therefore 
\[
DP_B(\delta)=
\min_{X\subseteq A(B),\,Y\subseteq I(B)}
\left\{
|Y|+\sum_{a\in A(B)} DP_a(\mathbf 1_{a\in X}):
H_B-S(D_{B,\delta,X,Y}) \text{ is chordal}
\right\}.
\]
For the root block $B_0$, the optimum value is
\[
\kappa(G,\C)=
\min_{X\subseteq A(B_0),\,Y\subseteq I(B_0)}
\left\{
|Y|+\sum_{a\in A(B_0)} DP_a(\mathbf 1_{a\in X}):
H_{B_0}-S(X\cup Y)\text{ is chordal}
\right\}.
\]

\medskip
\noindent
\textbf{Correctness:}
We prove by induction on the rooted block-cut tree that every state computes the claimed optimum.  For articulation states, once the status of $a$ is fixed, the subproblems on the child blocks interact only through that already fixed status. Hence they are independent, and the recurrence is exact:
\[
DP_a(\delta)=\delta+\sum_{B\in \mathcal B(a)} DP_B(\delta)
\]

For block states, let us fix a non-root block $B$ and a parent status $\delta$. Any globally feasible deletion set in
the subtree rooted at $B$ induces a unique choice of deleted child articulation colors $X\subseteq A(B)$, and  deleted internal colors $Y\subseteq I(B)$. The restriction of this deletion set to the colors of $B$ is exactly $D_{B,\delta,X,Y}$. By
Lemma~\ref{residual-block-local-lemma}, the residual augmented graph of the entire subtree rooted at $B$
is chordal if and only if \textit{(i)} the local residual graph $H_B - S(D_{B,\delta,X,Y})$  is chordal, and \textit{(ii)} for every child articulation color $a\in A(B)$, the residual augmented graph of the subtree rooted at $a$  is chordal under the induced deletion status of $a$. By the induction hypothesis, condition \textit{(ii)} is captured exactly by the terms $DP_a(\mathbf 1_{a\in X})$.
Thus every feasible global solution yields a feasible choice in the recurrence, and conversely every feasible choice in the recurrence, combined with optimal child solutions produces a globally feasible deletion set. Therefore the recurrence for $DP_B(\delta)$ is exact. Finally, the root formula is exact by the same reasoning.

\medskip
\noindent
\textbf{Cost:}
Every internal color of a block is charged exactly once, through a term $|Y|$ in that block state. Every articulation color is charged exactly once, namely through the $\delta$-term in its own articulation state. Hence no color is omitted and no color is double-counted.

\medskip
\noindent
\textbf{Running time:} Each block contains at most $b$ colors, so each block state enumerates at most $2^{|A(B)|+|I(B)|}\le 2^b$
local choices. For each such choice, the chordality test on $H_B-S(D_{B,\delta,X,Y})$ can be performed in polynomial time. The number of articulation and block states is linear in the size of
the block-cut tree, and therefore polynomial in $|\C|$. This yields total running time $2^{\mathcal{O}(b)}\cdot \mathrm{poly}(|E|+|\C|)$.
\end{proof}

\begin{corollary}
Let $\X$ be a clique-sum local hereditary graph class, and suppose that membership in $\X$ can be tested in polynomial time. If every block of the color-intersection graph $\Gamma$ has size at most $b$,
then the minimum number of colors whose deletion places the conflict graph in $\X$ can be computed exactly in time
\[
2^{O(b)}\cdot \mathrm{poly}(|E|+|\mathcal{C}|).
\]
\end{corollary}

\medskip
\noindent
\textit{Proof sketch:} The only place where chordality is
used in the proof of Theorem \ref{DP-bounded-block-theorem} is within the feasibility test to verify that $H_B-S(D_{B,\delta,X,Y})$  is chordal.
For a fixed clique-sum local hereditary class $\X$, Theorem \ref{cliquesumlocal-theorem} shows us that the  membership of the residual conflict graph in $\X$ is equivalent to testing the membership of all local residual graphs $H_B-S(D_{B,\delta,X,Y})$ in $\X$ thus we may replace the local chordality test by the local membership test for $\X$. But since by definition, we can test the membership in $\X$ in polynomial-time and the number of local choices per block is still at
most \(2^b\), the running-time analysis is unchanged. \qed


\medskip
A complementary algorithmic consequence of the same parameter does not rely on bounded block size at all: once a chordalizing color set is supplied we can easily solve the original optimization problem exactly by branching over the choices made inside the deleted color classes and solving the residual chordal instance. 

\begin{theorem}
\label{supplied-chordal-set}
Let $(G,\C)$ be an edge-colored graph, let $H$ be its conflict graph, and let $F\subseteq \C$ be a set of colors such that $H-S(F)$ is chordal. 
Then a maximum rainbow matching in $(G,\C)$ can be computed exactly in time
\[
\mathcal{O} \left(\prod_{\C_i\in F} (|\C_i|+1)\cdot (|V(H)|+|E(H)|)\right).
\]
In particular, if $|F|=k$ and every deleted color class has size at most $r$, then the running time is $ \mathcal{O} \left((r+1)^k\cdot (|V(H)|+|E(H)|)\right)$.
\end{theorem}

\begin{proof}
For each deleted color $\C_i\in F$, a feasible rainbow matching can use either no edge or exactly one edge from $\C_i$. We therefore branch over all choices
\[
\phi\in \prod_{\C_i\in F} (C_i\cup\{\bot\}),
\]
where $\phi(\C_i)=\bot$ means that no edge of color $\C_i$ is selected. The number of branches is exactly $\prod_{\C_i\in F} (|C_i|+1)$. Fix one branch $\phi$. First, we reject it if some pair of the chosen edges $\{\phi(c)\neq \bot\}$ are pairwise conflicting, since then no rainbow matching can realize this choice. Otherwise, let $M_\phi$ be the set of selected deleted-color edges, and restrict the residual instance to those vertices of $H-S(F)$ that are nonadjacent to every edge of $M_\phi$. Equivalently, we delete from $H-S(F)$ all vertices conflicting with the fixed choice $M_\phi$.

Because chordality is hereditary and $H-S(F)$ is chordal (by assumption), the remaining residual graph is chordal. Maximum stable set in a chordal graph can be computed in linear time. Hence, for the current branch, we can compute exactly the best compatible residual stable set $I_\phi$ in time $O(|V(H)|+|E(H)|)$, and the set $M_\phi\cup I_\phi$ is an optimal rainbow matching among all solutions realizing the deleted-color choices prescribed by $\phi$.

Taking the best solution over all branches yields an optimal rainbow matching for the original instance. The total running time is therefore
\[
\mathcal{O} \left(\prod_{\C_i\in F} (|\C_i|+1)\cdot (|V(H)|+|E(H)|)\right),
\]
and if $|F|=k$ and $|\C_i|\leq r$ for all $\C_i\in F$, this simplifies to $\mathcal{O} \left((r+1)^k\cdot (|V(H)|+|E(H)|)\right)$. 
\end{proof}

\begin{remark}
As before, the previous Theorem hold for any clique-sum local hereditary class $\X$ where the chordalizing set is replaced by any set $F$ such that $H-S(F) \in \X$, within exactly the same bounds, assuming that recognition to $\X$ can be done in polynomial time. 
\end{remark}

\begin{remark}
We note that, by using a standard backtracking argument, Theorem \ref{DP-bounded-block-theorem} can also give us the actual minimum  set $F^*$ such that $|F^*| = \kappa_\X(G,C)$ within the same time bound.
\end{remark}

\begin{corollary}
Combining the previous remark with Theorem \ref{supplied-chordal-set} gives us that Rainbow Matching is FPT when the target augmented residual graph $H$ belong to a clique-sum local class $\X$ with polynomial time recognition, the color-intersection graph $\Gamma$ consists of blocks of bounded size $b$ and the deleted color classes have size at most $r$. 
\end{corollary}

\section{Complexity}
In this section we tackle the complexity of finding the parameter $\kappa$. We firstly show that if the target hereditary class $\X$ admits a finite forbidden induced subgraph characterization of \emph{bounded size}  then the deletion problem becomes FPT. Then, we show that computing the parameter is \textbf{NP}-hard.

\begin{theorem}\label{fpt}
Let $\X$ be some hereditary graph class characterized by a finite family $\mathcal{F}_{\X}$ of forbidden induced subgraphs. Set $d_{\X} = \max\{|V(F)| : F \in \mathcal{F}_{\X}\}$. Then the following problem is FPT by $k$: does there exist a set of colors $F \subseteq \C$ with $|F| \leq k$ such that $H-S(F)$ belongs to $\X$?
More precisely, the problem can be solved in time $\mathcal{O}(|V(H)|^{d_\X}) + \mathcal{O}^*(d_\X^k)$.
\end{theorem}

The simple proof follows the typical pattern for similar vertex deletion problems \cite{DBLP:journals/ipl/Cai96,DBLP:journals/jcss/Abu-Khzam10} and reduces the problem to a bounded-arity hitting set instance on the color set: every forbidden induced subgraph gives rise to the set of colors appearing on it, and any valid deletion set must hit all such color signatures.

\begin{proof}
Since $\X$ is fixed, $\mathcal{F}_\X$ and $d_\X$ are also fixed as well. We will reduce our problem to a bounded-arity hitting set instance on the color set $C$. We construct a hypergraph $\mathcal{H}_\X$ with universe $C$ as follows: For each subset
$U \subseteq V(H)$ with $|U| \leq d_\X$, we test (in constant time) whether the induced subgraph $H[U]$ is isomorphic to some graph in $\mathcal{F}_\X$. If true, we add the hyperedge $c(U) =\{c(v): v\in U\}\subseteq C$ to $\mathcal{H}_\X$. Each hyperedge of  $\mathcal{H}_\X$ has size at most $d_\X$ and the whole construction is obviously polynomial since there are at most $O(|V(H)|^{d_\X})$ candidates sets. We will see that a set of colors $F \subseteq C$ is a solution for
\textit{Grouped Deletion to $\X$} if-f it is a hitting set of $\mathcal{H}_\X$.


($\Rightarrow$) Suppose that $F$ is a solution, i.e., $H-S(F)\in \X$ and let $c(U)$ be an hyperedge of $\mathcal{H}_\X$. By construction $H[U]$ is isomorphic to some forbidden induced subgraph in $\mathcal{F}_\X$. If $F \cap c(U)=\emptyset$  then obviously no vertex of $U$ is deleted when removing the colors in $F$, so the same induced subgraph $H[U]$ survives inside $H-S(F)$. This contradicts $H-S(F)\in \X$. Hence every hyperedge of $\mathcal{H}_\X$ is hit by $F$.

($\Leftarrow$) Conversely suppose that $F \subseteq C$ hits every hyperedge of $\mathcal{H}_\X$, but that $H-S(F)\notin \X$. Since $\X$ is hereditary and characterized by $\mathcal{F}_\X$, the residual graph
$H-S(F)$ must contain an induced subgraph $J$ isomorphic to some graph in $\mathcal{F}_\X$. Let $U=V(J) \subseteq V(H)$. Then $|U| \leq d_\X$, and by construction $c(U)$ is a hyperedge of $\mathcal{H}_\X$. But every vertex of $U$ survives in $H-S(F)$, so $F \cap c(U)=\emptyset$, contradicting that $F$ hits every hyperedge.

Since every hyperedge has size at most $d_\X$, we can solve the resulting bounded-arity hitting set instance by the standard branching algorithm: if some hyperedge is not yet hit, we can branch on one of the at most $d_\X$ colors. The previous search tree has depth at most $k$  and hence running time $\mathcal{O}^*(d_\X^k)$.
\end{proof}


\begin{theorem}\label{nphardness}
Given an edge-colored graph $(G,\C)$ and an integer $K$, deciding whether $\kappa(G,\C) \leq K$ is \textbf{NP}-complete.
\end{theorem}
\begin{proof}
Membership in NP is immediate. 
For NP-hardness, we reduce from \textsc{Vertex Cover}. Let $(G_{\mathrm{vc}}=(U,E_{\mathrm{vc}}),K)$ be an instance of \textsc{Vertex Cover}. 
We construct an edge-colored graph $(G,\C)$ as follows. For each vertex $u\in U$, introduce a \emph{public color}, also denoted by $u$. For each edge $e=\{u,v\}\in E_{\mathrm{vc}}$ and each index $i\in [K+1]$ we introduce two new \emph{private colors} $\alpha_{e,i},\beta_{e,i}$, distinct from any other colors. For each pair $(e,i)$, create four new vertices $x_{e,i}^1, x_{e,i}^2, x_{e,i}^3, x_{e,i}^4$ and add the 4-cycle $x_{e,i}^1x_{e,i}^2x_{e,i}^3x_{e,i}^4x_{e,i}^1$ to $G$. We color its four edges, in cyclic order, by $u,\alpha_{e,i},v,\beta_{e,i}$, where $e=\{u,v\}$. Denote these four gadget edges by $g_{e,i}^1=x_{e,i}^1x_{e,i}^2$, $g_{e,i}^2=x_{e,i}^2x_{e,i}^3$, $g_{e,i}^3=x_{e,i}^3x_{e,i}^4$, $g_{e,i}^4=x_{e,i}^4x_{e,i}^1$. All gadgets are pairwise vertex-disjoint. The construction is clearly polynomial.

Let $H$ be the conflict graph of the constructed instance. Recall that two vertices of $H$ are adjacent if and only if the corresponding edges of $G$ either share an endpoint or have the same color. It is straightforward to see that for every $e=\{u,v\}\in E_{\mathrm{vc}}$ and every $i\in [K+1]$, the four vertices of $H$ corresponding to $g_{e,i}^1,g_{e,i}^2,g_{e,i}^3,g_{e,i}^4$
induce a chordless cycle $C_4$. Indeed, consecutive gadget edges share an endpoint in the original 4-cycle, hence the corresponding conflict-graph vertices are adjacent. Opposite gadget edges are disjoint and have different colors, hence the corresponding conflict-graph vertices are nonadjacent. Therefore these four vertices induce
exactly a $C_4$.

We now prove that $G_{\mathrm{vc}}$  has a vertex cover of size at most  $K$ if-f $\kappa(G,\C)\leq K$. 

\smallskip
\noindent
$(\Rightarrow)$
Assume that $X\subseteq U$ is a vertex cover of $G_{\mathrm{vc}}$ with $|X|\leq K$. Delete exactly the public colors in $X$, and set $F =X$. We claim that $H-S(F)$ is chordal. Indeed fix a gadget corresponding to some edge $e=\{u,v\}\in E_{\mathrm{vc}}$. Since $X$ is a vertex cover, at least one of $u,v$ belongs to $X$. Hence at least one of the two public-color vertices of this gadget is deleted from the conflict graph. If exactly one of $u,v$ is deleted, then among the four gadget vertices only three remain, namely one public-color vertex and the two private-color
vertices. The surviving induced subgraph is a 3-vertex star $K_{1,2}$: the
surviving public-color vertex is adjacent to each surviving private vertex,
while the two private vertices are nonadjacent. If both $u$ and $v$ are deleted, then only the two private vertices remain, and they are isolated. Thus every surviving private vertex has degree at most $1$ in $H-S(F)$. In particular, every surviving private vertex is simplicial.

Now remove all surviving private vertices first. Since all gadgets are pairwise
vertex-disjoint, there are no shared-endpoint conflicts between different
gadgets. Moreover, different public colors are distinct, so there are no
same-color conflicts between different public colors. Hence after all private
vertices are removed, what remains is exactly the graph induced by the surviving
public-color vertices, and this graph is a disjoint union of cliques, one clique
for each undeleted public color.

A disjoint union of cliques is chordal. Since each removed private vertex was
simplicial, adding them back preserves chordality. Therefore $H-S(F)$ is
chordal. It follows that $\kappa(G,\C)\le |F|=|X|\leq K$.

\smallskip
\noindent
$(\Leftarrow)$
Assume that $F\subseteq \mathcal C$ is a set of colors with $|F|\leq K$ such
that $H-S(F)$ is chordal. Let $X=F\cap U$ i.e., $X$ is the set of deleted public colors. We claim that $X$ is a vertex cover of $G_{\mathrm{vc}}$. Indeed suppose not. Then there exists an edge $e=\{u,v\}\in E_{\mathrm{vc}}$ such that neither $u$ nor $v$ belongs to $X$. Equivalently, neither public color $u$ nor public color $v$ is deleted.

Consider the $K+1$ gadgets $Q_{e,1},Q_{e,2},\dots,Q_{e,K+1}$ associated with the edge $e$. Fix one such gadget $Q_{e,i}$. Since neither public color $u$ nor public color $v$ is deleted, both corresponding public-color vertices survive in $H-S(F)$. If neither private color $\alpha_{e,i}$ nor $\beta_{e,i}$ is deleted either, then all four gadget vertices survive, and by the claim above they still induce a chordless $C_4$ in the residual augmented graph. This is impossible because $H-S(F)$ is chordal.

Therefore, for each $i\in [K+1]$, at least one of $\alpha_{e,i}$ or $\beta_{e,i}$
must belong to $F$. Since all these private colors are distinct across the
$K+1$ copies, it follows that $|F|\geq K+1$, contradicting that $|F|\leq K$.

This contradiction shows that for every edge $e=\{u,v\}\in E_{\mathrm{vc}}$,
at least one of $u,v$ lies in $X$. Hence $X$ is a vertex cover of $G_{\mathrm{vc}}$. Finally, $|X|\leq |F|\leq K$, so $G_{\mathrm{vc}}$ has a vertex cover of size at most $K$.
\end{proof}

\section{Conclusion and Open Directions}
Our upper bound $\kappa+1$ (or $\kappa+r$ in the more general non-perfect case) is already tight at the first nontrivial point: there are simple instances with $\kappa=1$ (like the rainbow $C_5$) for which one round of Lasserre is not exact but two rounds suffice. On the other hand, we do not necessarily expect $\kappa+1$ to characterize the exact rank in full generality. This raises the following natural open questions. 

First, on which graph families does the deletion parameter provide a sharp bound of the Lasserre rank? Second, what is the exact parameterized complexity of computing the  deletion number $\kappa$? On the positive side, if the target hereditary class $\X$ is characterized by finitely many forbidden induced subgraphs of bounded size, the grouped deletion problem is FPT. This yields a broad tractable family that includes residual cluster, cograph, threshold, and split graphs among others. Beyond this family, however, the complexity should depend to the particular class we are aiming for. Our residual block-locality theorem makes chordal residual graphs the most natural next target for further exploration. We note that although chordal deletion in the standard sense is FPT \cite{DBLP:journals/algorithmica/Marx10}, we do not expect an easy adaptation of these techniques to be applicable in our group deletion to chordal. 
On the hardness side, perfect residual graphs are the most natural next target since they play a central role in the polyhedral side of our story, but their obstruction family is unbounded, and the grouped nature of the deletion operation makes W-hardness plausible.  

Another open direction is to identify in which other settings our parameter might play a role. This naturally leads us to the framework of \emph{partition-augmented independent sets}: given a graph $G$ and a partition $\mathcal{P} =\{V_1,\dots,V_m\}$ of $V(G)$, we form the augmented graph $G_{\mathcal{P}}$ by completing each part into a clique and we ask for a maximum independent set of $G_{\mathcal{P}}$.  This framework captures natural packing problems beyond Rainbow Matching, such as group-constrained set packings and multiple-choice interval scheduling and it connects very naturally to the already-studied broad framework of matroid-constrained stable sets from \cite{DBLP:conf/swat/FominGK024}. It would be interesting to understand whether the grouped-deletion and exactness  identified here extend to that setting as well.

\bibliographystyle{plainurl}
\bibliography{references}

@inproceedings{karlin2011integrality,
  title={Integrality gaps of linear and semi-definite programming relaxations for knapsack},
  author={Karlin, Anna R and Mathieu, Claire and Nguyen, C Thach},
  booktitle={International Conference on Integer Programming and Combinatorial Optimization},
  pages={301--314},
  year={2011},
  organization={Springer}
}

@article{rothvoss2013lasserre,
  title={The Lasserre hierarchy in approximation algorithms},
  author={Rothvo{\ss}, Thomas},
  journal={Lecture Notes for the MAPSP},
  pages={1--25},
  year={2013}
}

@inproceedings{DBLP:conf/ipco/Lasserre01,
  author       = {Jean B. Lasserre},
  editor       = {Karen I. Aardal and
                  Bert Gerards},
  title        = {An Explicit Exact {SDP} Relaxation for Nonlinear 0-1 Programs},
  booktitle    = {Integer Programming and Combinatorial Optimization, 8th International
                  {IPCO} Conference, Utrecht, The Netherlands, June 13-15, 2001, Proceedings},
  series       = {Lecture Notes in Computer Science},
  pages        = {293--303},
  publisher    = {Springer},
  year         = {2001},
  url          = {https://doi.org/10.1007/3-540-45535-3\_23},
  doi          = {10.1007/3-540-45535-3\_23},
  bibsource    = {dblp computer science bibliography, https://dblp.org}
}

@article{DBLP:journals/siamjo/Lasserre01,
  author       = {Jean B. Lasserre},
  title        = {Global Optimization with Polynomials and the Problem of Moments},
  journal      = {{SIAM} J. Optim.},
  volume       = {11},
  number       = {3},
  pages        = {796--817},
  year         = {2001},
  url          = {https://doi.org/10.1137/S1052623400366802},
  doi          = {10.1137/S1052623400366802},
  bibsource    = {dblp computer science bibliography, https://dblp.org}
}

@article{DBLP:journals/ejc/AharoniBCHS19,
  author       = {Ron Aharoni and
                  Eli Berger and
                  Maria Chudnovsky and
                  David M. Howard and
                  Paul D. Seymour},
  title        = {Large rainbow matchings in general graphs},
  journal      = {Eur. J. Comb.},
  volume       = {79},
  pages        = {222--227},
  year         = {2019},
  url          = {https://doi.org/10.1016/j.ejc.2019.01.012},
  doi          = {10.1016/J.EJC.2019.01.012},
  bibsource    = {dblp computer science bibliography, https://dblp.org}
}

@article{DBLP:journals/jacm/ItaiRT78,
  author       = {Alon Itai and
                  Michael Rodeh and
                  Steven L. Tanimoto},
  title        = {Some Matching Problems for Bipartite Graphs},
  journal      = {J. {ACM}},
  volume       = {25},
  number       = {4},
  pages        = {517--525},
  year         = {1978},
  url          = {https://doi.org/10.1145/322092.322093},
  doi          = {10.1145/322092.322093},
  bibsource    = {dblp computer science bibliography, https://dblp.org}
}

@article{DBLP:journals/jct/Drisko98,
  author       = {Arthur A. Drisko},
  title        = {Transversals in Row-Latin Rectangles},
  journal      = {J. Comb. Theory {A}},
  volume       = {84},
  number       = {2},
  pages        = {181--195},
  year         = {1998},
  url          = {https://doi.org/10.1006/jcta.1998.2894},
  doi          = {10.1006/JCTA.1998.2894},
  bibsource    = {dblp computer science bibliography, https://dblp.org}
}

@inproceedings{DBLP:conf/swat/FominGK024,
  author       = {Fedor V. Fomin and
                  Petr A. Golovach and
                  Tuukka Korhonen and
                  Saket Saurabh},
  editor       = {Hans L. Bodlaender},
  title        = {Stability in Graphs with Matroid Constraints},
  booktitle    = {19th Scandinavian Symposium and Workshops on Algorithm Theory, {SWAT}
                  2024, Helsinki, Finland, June 12-14, 2024},
  series       = {LIPIcs},
  pages        = {22:1--22:16},
  publisher    = {Schloss Dagstuhl - Leibniz-Zentrum f{\"{u}}r Informatik},
  year         = {2024},
  url          = {https://doi.org/10.4230/LIPIcs.SWAT.2024.22},
  doi          = {10.4230/LIPICS.SWAT.2024.22},
  bibsource    = {dblp computer science bibliography, https://dblp.org}
}

@inproceedings{DBLP:conf/soda/BessyBTW23,
  author       = {St{\'{e}}phane Bessy and
                  Marin Bougeret and
                  Dimitrios M. Thilikos and
                  Sebastian Wiederrecht},
  editor       = {Nikhil Bansal and
                  Viswanath Nagarajan},
  title        = {Kernelization for Graph Packing Problems via Rainbow Matching},
  booktitle    = {Proceedings of the 2023 {ACM-SIAM} Symposium on Discrete Algorithms,
                  {SODA} 2023, Florence, Italy, January 22-25, 2023},
  pages        = {3654--3663},
  publisher    = {{SIAM}},
  year         = {2023},
  url          = {https://doi.org/10.1137/1.9781611977554.ch139},
  doi          = {10.1137/1.9781611977554.CH139},
  bibsource    = {dblp computer science bibliography, https://dblp.org}
}

@article{chvatal1975certain,
  title={On certain polytopes associated with graphs},
  author={Chv{\'a}tal, Va{\v{s}}ek},
  journal={Journal of Combinatorial Theory, Series B},
  volume={18},
  number={2},
  pages={138--154},
  year={1975},
  publisher={Academic Press}
}

@article{lovasz1991cones,
  title={Cones of matrices and set-functions and 0--1 optimization},
  author={Lov{\'a}sz, L{\'a}szl{\'o} and Schrijver, Alexander},
  journal={SIAM journal on optimization},
  volume={1},
  number={2},
  pages={166--190},
  year={1991},
  publisher={SIAM}
}

@article{laurent2003comparison,
  title={A comparison of the Sherali-Adams, Lov{\'a}sz-Schrijver, and Lasserre relaxations for 0--1 programming},
  author={Laurent, Monique},
  journal={Mathematics of Operations Research},
  volume={28},
  number={3},
  pages={470--496},
  year={2003},
  publisher={INFORMS}
}

@article{thetabody,
author = {Gouveia, Jo\~{a}o and Parrilo, Pablo A. and Thomas, Rekha R.},
title = {Theta Bodies for Polynomial Ideals},
journal = {SIAM Journal on Optimization},
volume = {20},
number = {4},
pages = {2097-2118},
year = {2010},
doi = {10.1137/090746525},

URL = { 
    
        https://doi.org/10.1137/090746525

},
eprint = { 
    
        https://doi.org/10.1137/090746525
}

}

@article{DBLP:journals/mp/BianchiENT17,
  author       = {Silvia M. Bianchi and
                  Mariana S. Escalante and
                  Graciela L. Nasini and
                  Levent Tun{\c{c}}el},
  title        = {Lov{\'{a}}sz-Schrijver SDP-operator, near-perfect graphs and
                  near-bipartite graphs},
  journal      = {Math. Program.},
  volume       = {162},
  number       = {1-2},
  pages        = {201--223},
  year         = {2017},
  url          = {https://doi.org/10.1007/s10107-016-1035-1},
  doi          = {10.1007/S10107-016-1035-1},
  bibsource    = {dblp computer science bibliography, https://dblp.org}
}

@article{DBLP:journals/disopt/BienstockO04,
  author       = {Daniel Bienstock and
                  Nuri {\"{O}}zbay},
  title        = {Tree-width and the Sherali-Adams operator},
  journal      = {Discret. Optim.},
  volume       = {1},
  number       = {1},
  pages        = {13--21},
  year         = {2004},
  url          = {https://doi.org/10.1016/j.disopt.2004.03.002},
  doi          = {10.1016/J.DISOPT.2004.03.002},
  bibsource    = {dblp computer science bibliography, https://dblp.org}
}

@techreport{wainwright2004treewidth,
  title={Treewidth-based conditions for exactness of the Sherali-Adams and Lasserre relaxations},
  author={Wainwright, Martin J and Jordan, Michael I},
  year={2004},
  institution={Technical Report 671, University of California, Berkeley}
}

@article{DBLP:journals/siamcomp/CarbonnelRZ22,
  author       = {Cl{\'{e}}ment Carbonnel and
                  Miguel Romero and
                  Stanislav Zivn{\'{y}}},
  title        = {The Complexity of General-Valued Constraint Satisfaction Problems
                  Seen from the Other Side},
  journal      = {{SIAM} J. Comput.},
  volume       = {51},
  number       = {1},
  pages        = {19--69},
  year         = {2022},
  url          = {https://doi.org/10.1137/19m1250121},
  doi          = {10.1137/19M1250121},
  bibsource    = {dblp computer science bibliography, https://dblp.org}
}

@article{pfender2014complexity,
 author       = {Van Bang Le and
                  Florian Pfender},
  title        = {Complexity results for rainbow matchings},
  journal      = {Theor. Comput. Sci.},
  volume       = {524},
  pages        = {27--33},
  year         = {2014},
  url          = {https://doi.org/10.1016/j.tcs.2013.12.013},
  doi          = {10.1016/J.TCS.2013.12.013},
  bibsource    = {dblp computer science bibliography, https://dblp.org}
}

@article{gupta2019parameterized,
  title={Parameterized algorithms and kernels for rainbow matching},
  author={Gupta, Sushmita and Roy, Sanjukta and Saurabh, Saket and Zehavi, Meirav},
  journal={Algorithmica},
  volume={81},
  number={4},
  pages={1684--1698},
  year={2019},
  publisher={Springer}
}

@article{DBLP:journals/algorithmica/GuptaRSZ20,
  author       = {Sushmita Gupta and
                  Sanjukta Roy and
                  Saket Saurabh and
                  Meirav Zehavi},
  title        = {Quadratic Vertex Kernel for Rainbow Matching},
  journal      = {Algorithmica},
  volume       = {82},
  number       = {4},
  pages        = {881--897},
  year         = {2020},
  url          = {https://doi.org/10.1007/s00453-019-00618-0},
  doi          = {10.1007/S00453-019-00618-0},
  bibsource    = {dblp computer science bibliography, https://dblp.org}
}

@article{hperfect,
  author       = {Najiba Sbihi and
                  J. P. Uhry},
  title        = {A class of h-perfect graphs},
  journal      = {Discret. Math.},
  volume       = {51},
  number       = {2},
  pages        = {191--205},
  year         = {1984},
  url          = {https://doi.org/10.1016/0012-365X(84)90071-2},
  doi          = {10.1016/0012-365X(84)90071-2},
  bibsource    = {dblp computer science bibliography, https://dblp.org}
}

@inproceedings{DBLP:conf/mfcs/Stamoulis14,
  author       = {Georgios Stamoulis},
  editor       = {Erzs{\'{e}}bet Csuhaj{-}Varj{\'{u}} and
                  Martin Dietzfelbinger and
                  Zolt{\'{a}}n {\'{E}}sik},
  title        = {Approximation Algorithms for Bounded Color Matchings via Convex Decompositions},
  booktitle    = {Mathematical Foundations of Computer Science 2014 - 39th International
                  Symposium, {MFCS} 2014, Budapest, Hungary, August 25-29, 2014. Proceedings,
                  Part {II}},
  series       = {Lecture Notes in Computer Science},
  pages        = {625--636},
  publisher    = {Springer},
  year         = {2014},
  url          = {https://doi.org/10.1007/978-3-662-44465-8\_53},
  doi          = {10.1007/978-3-662-44465-8\_53},
  bibsource    = {dblp computer science bibliography, https://dblp.org}
}

@article{DBLP:journals/tcs/MastrolilliS14,
  author       = {Monaldo Mastrolilli and
                  Georgios Stamoulis},
  title        = {Bi-criteria and approximation algorithms for restricted matchings},
  journal      = {Theor. Comput. Sci.},
  volume       = {540},
  pages        = {115--132},
  year         = {2014},
  url          = {https://doi.org/10.1016/j.tcs.2013.11.027},
  doi          = {10.1016/J.TCS.2013.11.027},
  bibsource    = {dblp computer science bibliography, https://dblp.org}
}

@article{DBLP:journals/disopt/KelkS19,
  author       = {Steven Kelk and
                  Georgios Stamoulis},
  title        = {Integrality gaps for colorful matchings},
  journal      = {Discret. Optim.},
  volume       = {32},
  pages        = {73--92},
  year         = {2019},
  url          = {https://doi.org/10.1016/j.disopt.2018.12.003},
  doi          = {10.1016/J.DISOPT.2018.12.003},
  bibsource    = {dblp computer science bibliography, https://dblp.org}
}

@article{DBLP:journals/ipl/Cai96,
  author       = {Leizhen Cai},
  title        = {Fixed-Parameter Tractability of Graph Modification Problems for Hereditary
                  Properties},
  journal      = {Inf. Process. Lett.},
  volume       = {58},
  number       = {4},
  pages        = {171--176},
  year         = {1996},
  url          = {https://doi.org/10.1016/0020-0190(96)00050-6},
  doi          = {10.1016/0020-0190(96)00050-6},
  bibsource    = {dblp computer science bibliography, https://dblp.org}
}

@article{DBLP:journals/jcss/Abu-Khzam10,
  author       = {Faisal N. Abu{-}Khzam},
  title        = {A kernelization algorithm for d-Hitting Set},
  journal      = {J. Comput. Syst. Sci.},
  volume       = {76},
  number       = {7},
  pages        = {524--531},
  year         = {2010},
  url          = {https://doi.org/10.1016/j.jcss.2009.09.002},
  doi          = {10.1016/J.JCSS.2009.09.002},
  bibsource    = {dblp computer science bibliography, https://dblp.org}
}

@article{DBLP:journals/tcs/HeggernesHJKV13,
  author       = {Pinar Heggernes and
                  Pim van 't Hof and
                  Bart M. P. Jansen and
                  Stefan Kratsch and
                  Yngve Villanger},
  title        = {Parameterized complexity of vertex deletion into perfect graph classes},
  journal      = {Theor. Comput. Sci.},
  volume       = {511},
  pages        = {172--180},
  year         = {2013},
  url          = {https://doi.org/10.1016/j.tcs.2012.03.013},
  doi          = {10.1016/J.TCS.2012.03.013},
  bibsource    = {dblp computer science bibliography, https://dblp.org}
}

@article{DBLP:journals/algorithmica/Marx10,
  author       = {D{\'{a}}niel Marx},
  title        = {Chordal Deletion is Fixed-Parameter Tractable},
  journal      = {Algorithmica},
  volume       = {57},
  number       = {4},
  pages        = {747--768},
  year         = {2010},
  url          = {https://doi.org/10.1007/s00453-008-9233-8},
  doi          = {10.1007/S00453-008-9233-8},
  bibsource    = {dblp computer science bibliography, https://dblp.org}
}

@article{DBLP:journals/mst/MorawietzGKS22,
  author       = {Nils Morawietz and
                  Niels Gr{\"{u}}ttemeier and
                  Christian Komusiewicz and
                  Frank Sommer},
  title        = {Refined Parameterizations for Computing Colored Cuts in Edge-Colored
                  Graphs},
  journal      = {Theory Comput. Syst.},
  volume       = {66},
  number       = {5},
  pages        = {1019--1045},
  year         = {2022},
  url          = {https://doi.org/10.1007/s00224-022-10101-z},
  doi          = {10.1007/S00224-022-10101-Z},
  bibsource    = {dblp computer science bibliography, https://dblp.org}
}

@inproceedings{DBLP:conf/mfcs/LafondL23,
  author       = {Manuel Lafond and
                  Weidong Luo},
  editor       = {J{\'{e}}r{\^{o}}me Leroux and
                  Sylvain Lombardy and
                  David Peleg},
  title        = {Parameterized Complexity of Domination Problems Using Restricted Modular
                  Partitions},
  booktitle    = {48th International Symposium on Mathematical Foundations of Computer
                  Science, {MFCS} 2023, Bordeaux, France, August 28 - September 1, 2023},
  series       = {LIPIcs},
  pages        = {61:1--61:14},
  publisher    = {Schloss Dagstuhl - Leibniz-Zentrum f{\"{u}}r Informatik},
  year         = {2023},
  url          = {https://doi.org/10.4230/LIPIcs.MFCS.2023.61},
  doi          = {10.4230/LIPICS.MFCS.2023.61},
  bibsource    = {dblp computer science bibliography, https://dblp.org}
}

\end{document}